\documentclass[a4paper,11pt]{article}
\pdfoutput=1 
\usepackage{jheppub} 
\usepackage{amsmath,amsfonts,amssymb,amsthm}
\usepackage[british]{babel}
\usepackage{bbm,amscd,bbold}
\usepackage{array}
\usepackage[svgnames,table]{xcolor}
\usepackage[T1]{fontenc}
\usepackage{inputenc}
\usepackage{lmodern} 
\usepackage{mathrsfs}
\usepackage{verbatim}
\usepackage{pdfsync}
\usepackage[final]{microtype}
\usepackage{adjustbox}
\usepackage{booktabs} 
\usepackage{caption,enumitem}
\usepackage{appendix}
\usepackage{chngcntr}
\usepackage{apptools}
\usepackage{romanbar}
\usepackage{tikz,pgfplots}
\pgfplotsset{compat=1.18}
\usepgfplotslibrary{fillbetween}
\usetikzlibrary{calc,arrows,arrows.meta,cd,shapes.geometric,positioning,through,decorations.markings,decorations.text}

\hypersetup{
}

\usepackage[hyperpageref]{backref}
\renewcommand*{\backref}[1]{}
\renewcommand*{\backrefalt}[4]{%
  \ifcase #1%
  \or [Page~#2.]%
  \else [Pages~#2.]%
  \fi%
}
\theoremstyle{plain}
\newtheorem{lemma}{Lemma}

\theoremstyle{definition}

\newcommand{\g}{\mathfrak{g}}

\renewcommand{\a}{\mathfrak{a}}
\renewcommand{\b}{\mathfrak{b}}

\renewcommand{\d}{\partial}

\renewcommand{\r}{\mathfrak{r}}

\newcommand{\so}{\mathfrak{so}}

\renewcommand{\t}{\mathfrak{t}}
\newcommand{\fm}{\mathfrak{m}}

\newcommand{\be}{\boldsymbol{e}}
\newcommand{\bu}{\boldsymbol{u}}
\newcommand{\x}{\boldsymbol{x}}

\newcommand{\bj}{\boldsymbol{j}}
\newcommand{\p}{\boldsymbol{p}}
\newcommand{\bk}{\boldsymbol{k}}
\newcommand{\ba}{\boldsymbol{a}}

\newcommand{\bzero}{\boldsymbol{0}}
\newcommand{\bbeta}{\boldsymbol{\beta}}
\newcommand{\bpi}{\boldsymbol{\pi}}
\newcommand{\bdelta}{\boldsymbol{\delta}}
\newcommand{\bd}{\boldsymbol{d}}
\newcommand{\bP}{\boldsymbol{P}}
\newcommand{\bD}{\boldsymbol{D}}
\newcommand{\sg}{\boldsymbol{\mathsf{g}}}

\newcommand{\sfb}{\boldsymbol{\mathsf{b}}}
\newcommand{\sfc}{\boldsymbol{\mathsf{c}}}
\newcommand{\sA}{\boldsymbol{\mathsf{A}}}
\newcommand{\sM}{\boldsymbol{\mathsf{M}}}

\newcommand{\scrO}{\mathscr{O}}

\newcommand{\eH}{\mathcal{H}}

\newcommand{\eO}{\mathcal{O}}

\newcommand{\eB}{\mathcal{B}}

\newcommand{\ann}{\operatorname{ann}}

\renewcommand{\Re}{\operatorname{Re}}
\renewcommand{\Im}{\operatorname{Im}}
\newcommand{\ad}{\operatorname{ad}}

\newcommand{\Ad}{\operatorname{Ad}}

\newcommand{\Tr}{\operatorname{Tr}}

\newcommand{\RR}{\mathbb{R}}

\newcommand{\QQ}{\mathbb{Q}}
\newcommand{\CC}{\mathbb{C}}

\newcommand{\GL}{\operatorname{GL}}

\newcommand{\SO}{\operatorname{SO}}
\newcommand{\SU}{\operatorname{SU}}
\newcommand{\U}{\operatorname{U}}

\newcommand{\Spin}{\operatorname{Spin}}

\definecolor{dkgr}{rgb}{0,0.6,0}
\definecolor{gris}{rgb}{0.5,0.5,0.5}

\usepackage{pdflscape}

\usepackage{braket} 
\usepackage{bm} 

\providecommand*{\pd}{\partial}
\renewcommand*{\pd}{\partial}

\providecommand*{\db}{{\bm{d}}}
\renewcommand*{\db}{{\bm{d}}}

\providecommand*{\zb}{{\bm{z}}}
\renewcommand*{\zb}{{\bm{z}}}

\providecommand*{\pib}{{\bm{\pi}}}
\renewcommand*{\pib}{{\bm{\pi}}}

\newcommand{\bv}{\boldsymbol{v}}

\title{\boldmath Planons and their Carroll--Galilei symmetries}

 \author[a,1]{José Figueroa-O'Farrill,\note{ORCID: \href{https://orcid.org/0000-0002-9308-9360}{0000-0002-9308-9360}}}
 \author[b,c,2]{Simon Pekar,\note{ORCID: \href{https://orcid.org/0000-0002-0765-8986}{0000-0002-0765-8986}}}
 \author[d,e,g,3]{Alfredo Pérez,\note{ORCID: \href{https://orcid.org/0000-0003-0989-9959}{0000-0003-0989-9959}}}
 \author[f,g,4]{and Stefan Prohazka\note{ORCID: \href{https://orcid.org/0000-0002-3925-3983}{0000-0002-3925-3983}}}

\affiliation[a]{Maxwell Institute and School of Mathematics, The University of Edinburgh, James Clerk Maxwell Building, Peter Guthrie Tait Road,
  Edinburgh EH9 3FD, Scotland, United Kingdom}
\affiliation[b]{International School for Advanced Studies (SISSA), Via Bonomea 265, 34136 Trieste, Italy}
\affiliation[c]{Istituto Nazionale di Fisica Nucleare, Sezione di Trieste, Via Valerio 2, 34127 Trieste, Italy}
\affiliation[d]{Centro de Estudios Científicos (CECs), Avenida Arturo Prat 514, Valdivia, Chile}
\affiliation[e]{Facultad de Ingeniería, Universidad San Sebastián, sede Valdivia, General Lagos 1163, Valdivia 5110693, Chile}
\affiliation[f]{University of Vienna, Faculty of Physics, Mathematical Physics, Boltzmanngasse 5, 1090, Vienna, Austria}
\affiliation[g]{Erwin Schr\"odinger Int.\ Institute for Mathematics and Physics, University of Vienna}

 \emailAdd{j.m.figueroa@ed.ac.uk}
 \emailAdd{spekar@sissa.it}
 \emailAdd{alfredo.perez@uss.cl}
 \emailAdd{stefan.prohazka@univie.ac.at}

 \abstract{We study the dynamics of planons, particles whose mobility
   is restricted to a plane, through the classification of coadjoint
   orbits and unitary irreducible representations of the centrally
   extended planon group.  Planons are closely related to
   Galilei/Bargmann symmetries and, remarkably, the often-ignored
   massless coadjoint orbits of the Galilei group play a central rôle
   in their description.  We thereby provide a nontrivial physical
   interpretation of these orbits.  We further construct classical and
   quantum dipoles as bound states of monopoles, where the restricted
   planar motion arises from a novel mixed Carroll--Galilei symmetry.
   We also argue that the simplest and already experimentally realised
   systems with Carroll symmetry are in crystals.}
  
\begin{document}
\maketitle

\section{Introduction}
\label{sec:introduction}

The Galilei algebra and its centrally extended version, the Bargmann
algebra, play a central role in physics. They describe the
kinematics of particles moving at velocities much smaller than the
speed of light. From the perspective of group theory, such
non-relativistic particles are described by coadjoint orbits or, in
the quantum case, by unitary irreducible representations of the
Bargmann group, where the mass is non-vanishing
\cite{Bargmann:1954gh}. Surprisingly, the Galilei/Bargmann group
admits various other coadjoint orbits and unitary irreducible
representations (see, e.g., \cite{Figueroa-OFarrill:2024ocf} for a complete
classification), most of which lack a clear physical interpretation,
even though they are mathematically well-defined (see, e.g.,
\cite{Duval:2005ry} for a particular application to geometrical
optics).

In this work, we show that the coadjoint orbits and unitary
irreducible representations of the Bargmann group, which are often
ignored in applications in nonrelativistic physics, play a key role in
the description of certain systems with restricted mobility. These so-called planons belong to the class of fracton
models~\cite{Nandkishore:2018sel,Pretko:2020cko,Grosvenor:2021hkn,Gromov:2022cxa}.
Furthermore, systems with these symmetries are, via duality to
elasticity in $2+1$ dimensions, experimentally
realised~\cite{Pretko:2017kvd}.\footnote{The relation is between
  fractons and crystalline defects. The immovable fractons (or
  monopoles) correspond to disclinations, whereas the partially
  movable dipoles correspond to dislocations. The allowed motion which
  is restricted to be orthogonal to the dipole moment is then related
  to the glide constraint in crystals which allows dislocations to
  move only in the direction of the Burgers vector.}

Systems with immobility or restricted mobility have shown up at
various corners of physics with possible applications ranging from
quantum
memory~\cite{Nandkishore:2018sel,Pretko:2020cko,Grosvenor:2021hkn,Gromov:2022cxa},
to (conformal) carrollian symmetries and their relation to
asymptotically flat
spacetimes~\cite{Duval:2014uva,Donnay:2022aba,Bagchi:2022emh,Figueroa-OFarrill:2021sxz,Have:2024dff}.
Due to the restricted mobility, standard quantum field theories results
often do not apply. This therefore presents novel challenges and
consequently new opportunities to sharpen our understanding of quantum
field theories~\cite{Brauner:2022rvf,Cordova:2022ruw}.  One intriguing
question arises: \textit{What defines these novel particles with
  restricted mobility, and how can they be coupled to other systems?}

In this work we answer this question for particles with conserved
electric charge, but what might be less familiar, also conserved
dipole moment and trace (of quadrupole) charge
\begin{align}
  \label{eq:charge-def}
  Q&= \int d^{3}x \rho  & \bm{D}&=\int d^{3}x \rho \x  &  Z&=\frac{1}{2}\int d^{3}x\,\rho\left\Vert \boldsymbol{x}\right\Vert ^{2} .
\end{align}
These charges can be derived as a consequence of the following current
conservation and suitable fall-offs for the fields towards spatial
infinity
\begin{align}
 \dot \rho + \pd_{i}\pd_{j} J^{ij} &= 0  & \delta_{ij}J^{ij} = 0 \, .
\end{align}

To build some intuition let us first consider a charged isolated
monopole with trajectory $\bm{z}(t)$ and charge density
$\rho=q\delta(\bm{x}-\bm{z}(t))$. Since $q$ is constant, the charge
$Q=q$ is trivially conserved, however the conservation of the dipole
moment, $\dot{\bm{D}} = q \dot{\bm{z}}(t)=\bm{0}$, puts
restrictions on the mobility. Hence, isolated monopoles are immobile
which is the characteristic feature of
fractons~\cite{Nandkishore:2018sel,Pretko:2020cko,Grosvenor:2021hkn}.

What about dipoles? The charge density of a dipole with a constant
dipole moment $\bm{d}$ is given by
$\rho=-d^{i}\pd_{i}\delta(\boldsymbol{x}-\boldsymbol{z}(t))$ and leads
to the following charges
\begin{align}
  \label{eq:charge-dip}
  Q&= 0  & \bm{D}&=\bm{d}  &  Z&=\bd\cdot\boldsymbol{z}(t) \, .
\end{align}
As expected, the charge vanishes and the dipole moment is conserved,
but the trace charge puts restrictions on the mobility
\begin{equation}
  \label{eq:mob_rest_dipoles}
  \dot Z = \bd\cdot\dot{\boldsymbol{z}}(t) =0 \, .
\end{equation}
From this condition, it is clear that the conservation of the trace
charge implies that the velocity along the direction of the dipole
moment must vanish. Thus the dipole can move only in the plane
orthogonal to its dipole moment, maintaining its orientation, see
Figure~\ref{fig:planon}. This distinctive characteristic feature
justifies why they are referred to as ``planons.'' As we will explain,
planons have an interesting mixed Carroll--Galilei
symmetry. Intuitively the particle Lagrangian is invariant under
galilean boosts in the transverse direction, where it can move freely,
and under carrollian boosts in the direction in which its motion is
restricted.

\begin{figure}[h]
  \centering
  
\begin{tikzpicture}[scale=3.2]

\draw[->] (0,0,0) -- (0,0,1.5) node[anchor=south east] {$x_1$};
\draw[->] (0,0,0) -- (1.8,0,0) node[anchor=north east] {$x_2$};
\draw[->] (0,0,0) -- (0,1.8,0) node[anchor=north west] {$x_3$};

\fill[blue!20,opacity=0.5] (1.5,0,0) -- (0,1.5,0) -- (0,1.5,1.5) -- (1.5,0,1.5) -- cycle;

\draw[blue!50] (1.5,0,0) -- (0,1.5,0);
\draw[blue!50] (0,1.5,0) -- (0,1.5,1.5);
\draw[blue!50] (0,1.5,1.5) -- (1.5,0,1.5);
\draw[blue!50] (1.5,0,1.5) -- (1.5,0,0);

\draw[->,gray!80] (0,0,0) -- (0.5,0.8,0.5) node[anchor=east] {$\zb(t_1)$};
\draw[->,gray!80] (0,0,0) -- (1,0.3,0.5) node[anchor=south east] {$\zb(t_2)$};

\draw[-Stealth,red,thick,line cap=round] (0.5,0.8,0.5) -- (0.5,0.8,0.1) node[anchor=north] {$\db$};

\draw[-Stealth,red,thick,line cap=round] (1,0.3,0.5) -- (1,0.3,0.1) node[anchor=north] {$\db$};

\end{tikzpicture}
\caption{Due to the conservation of the trace charge
  $\dot Z = \bd\cdot\dot{\boldsymbol{z}}\left(t\right)=0$ planons with
  dipole moment $\db$ are allowed to move but are restricted to the
  plane transversal to the dipole moment.}
  \label{fig:planon}
\end{figure}
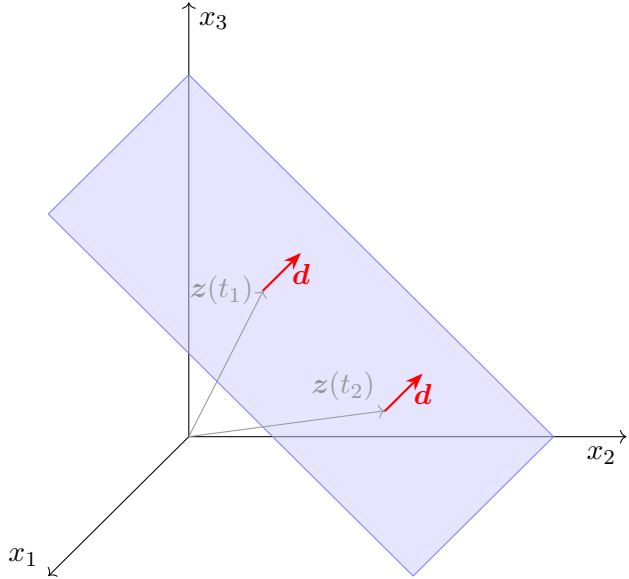

The whole system also possesses rotational, time and spatial translational
symmetries, whose generators we denote by $L_{ab}$, $H$, $P_{a}$. The
planon algebra, which is an instance of the multipole
algebra~\cite{Gromov:2018nbv}, is then spanned by
$\{L_{ab}, P_a, Q, D_a, Z := \delta^{ab}Q_{ab}, H\}$.  There are no
boost generators, since like for fractons, the geometry is
aristotelian~\cite{Bidussi:2021nmp,Jain:2021ibh}. The non-vanishing
brackets are 
\begin{subequations}
  \label{eq:planonalg}
\begin{align}
[L_{ab},L_{cd}] &= \delta_{bc}L_{ad} - \delta_{ac}L_{bd} - \delta_{bd}L_{ac} + \delta_{ad}L_{bc} \\
[L_{ab},P_c] &= \delta_{bc} P_a - \delta_{ac} P_b \\
[L_{ab},D_c] &= \delta_{bc} D_a - \delta_{ac} D_b \\
[P_a,Z] &= D_a \\
[P_a,D_b] &= \delta_{ab} Q \,,
\end{align}
\end{subequations}
where $H$ is central. As already emphasised by
Gromov~\cite{Gromov:2018nbv}, this algebra is closely related to the
centrally extended Galilei algebra, i.e., the Bargmann algebra.  We
have summarised their correspondence in
Appendix~\ref{sec:bargmann-to-planon-summary} and
Table~\ref{tab:planon-barg}, and at various instances we will use or
contrast our results to this case~\cite{Figueroa-OFarrill:2024ocf}.

The classical elementary particles of a Lie group $G$ are its
homogeneous symplectic manifolds. Roughly speaking, these are the
coadjoint orbits of one-dimensional central extensions of $G$.  In
Section~\ref{sec:centr-extens-plan} we exhibit the unique (nontrivial
and up to equivalence) one-dimensional central extension of the planon
group.  At the Lie algebra level, this consists in adding a new
generator $W$ to the above basis for the planon algebra and modifying
the Lie brackets in equation~\eqref{eq:planonalg} by the addition of
\begin{equation}
  \label{eq:into-new-central-bracket}
  [H,Z] = W \, .
\end{equation}
In Section~\ref{sec:coadjoint-orbits} we study the coadjoint
representation of the centrally extended planon group or,
equivalently, the action of the centrally extended planon group on the
conserved charges.  We then restrict our attention to $3$ spatial
dimensions and classify the coadjoint orbits in
Section~\ref{sec:coadjoint-orbits-n=3}.  These are summarised in
Table~\ref{tab:planon-coad-orbs}.

In Section~\ref{sec:actions-mobility} we discuss the mobility of
the classical particles.  We start in Section~\ref{sec:MC-one-form}
with a review of how to associate a particle trajectory with a
coadjoint orbit.  We use the data defining the coadjoint orbit to
define a variational problem for curves in the Lie group: curves which
we can then map to particle trajectories on a chosen homogeneous
spacetime of the Lie group.  This is done for each of the coadjoint
orbits in Section~\ref{sec:particle-dynamics} and the results are
summarised in Table~\ref{tab:stabilisers}.  We then concentrate on
those coadjoint orbits of the unextended planon group (the other
orbits lead to trajectories where energy is unbounded from below) and
study their dynamical systems in more detail in
Section~\ref{sec:particle-dynamics-orbitwise}.  We indeed find a
monopole (and its spinning counterpart) from first principles
\begin{equation}
  \label{eq:monopole_act}
  S_{\mathrm{mono}}\left[\boldsymbol{x},\boldsymbol{\pi}\right]=\int dt  [\boldsymbol{\pi}\cdot\dot{\boldsymbol{x}}] \, .
\end{equation}
As expected, the equation of motion restricts the mobility to
$\dot{\bm{x}} =\bm{0}$. It is interesting to note that this particle
is related to the well-known massive Galilei particle, where the
immobility translates into the constancy of velocity (see
Section~\ref{sec:orbit-6-monopole} for details).  We also find
elementary dipoles with movement restricted to a plane. They are
related to the often ignored massless Galilei orbits.

In Section~\ref{sec:composite-dipoles} we construct composite dipoles
from the bound state of two monopoles, leading to the action
\begin{align}
  \label{eq:Spl}
  S_{\text{dip}}[\bm{z},\pib,\boldsymbol{\sigma},\boldsymbol{d}] &=\int dt 
  \left[
    \pi_{i}\dot{z}^{i}+\sigma_{i}\dot{d}^{i}-H
  \right] &
  H &=\frac{1}{2m}P_{ij}\pi^{i}\pi^{j}
\end{align}
where $P_{ij}=\delta_{ij}-\hat d_{i}\hat d_{j}$ is the projector
transverse to the dipole moment, and $\hat{\bm{d}}$ is the unit vector
in the direction of the dipole moment.  This system has the expected
planon symmetries and (im)mobilty, which can be seen as a consequence
of a novel mixed Carroll--Galilei symmetry.  In the transverse
direction, the particle admits Galilei boosts and can therefore freely
move. In the longitudinal direction, it has a Carroll boost symmetry
and is therefore stuck, since isolated Carroll
particles~\cite{Duval:2014uoa,Bergshoeff:2014jla} like
fractons~\cite{Figueroa-OFarrill:2023vbj} cannot move.

Following~\cite{Perez:2023uwt}, this particle action can be
consistently coupled to the traceless scalar gauge
theory~\cite{Pretko:2016lgv}
\begin{equation}
  S_{\text{int}}[\bm{z},\bm{d}]=\int dt d^{3}x\,\left[\phi\rho-A_{ij}J^{ij}\right]
  =\int dt\,\left[\partial_{i}\phi(t,\boldsymbol{z}(t))d^{i}+A_{ij}(t,\boldsymbol{z}(t))d^{\left\langle i\right.}\dot{z}^{\left.j\right\rangle }\right]  \, .
\end{equation}
where $\phi$ represents the scalar potential, while $A_{ij}$ is the
symmetric and traceless tensor potential. Here and in the following, angle brackets will implement the symmetric and traceless projection. The variation of the coupled
system leads to the generalised Lorentz force law presented by
Pretko~\cite{Pretko:2016lgv}, but now follows from an action principle
including a planon particle.

In Section~\ref{sec:first-quantisation}, we discuss the quantisation of the composite
model and obtain a gaussian Schr\"odinger-like theory described by the field
$\psi(t,\bm{x})$, with action
\begin{align}
  S[\psi,\psi^{*}] &=\int dt d^{3}x
                     \left[
                     i\psi^{*} \mathcal D_{t}\psi - \frac{1}{2m} P^{ij} \mathcal{D}_{i}\psi (\mathcal{D}_{j}\psi)^{*}
                     \right],
\end{align}
where we have introduced
$\mathcal{D}_t= \partial_t - i d^i \partial_i\phi$ and
$\mathcal{D}_{i}=i\pd_{i}+A_{ik}d^{k}$ with $\bm{d}$ being a constant
parameter. This model already appeared 
in~\cite{Kumar_2019,Pretko:2019omh} (see
also~\cite[II.B.3]{Gromov:2022cxa}) in the context of
fracton/elasticity duality. The uncoupled theory once again exhibits the Carroll--Galilei symmetry  previously discussed (see
Section~\ref{sec:summ-plan-field}).

In Section~\ref{sec:quant-plan-part} we study the planon elementary
quantum systems, i.e., the unitary irreducible representations (UIRs)
of the centrally extended planon group.  We use the fact that the
centrally extended planon group can be written as a semidirect product
$B \ltimes A$ with $B$ isomorphic to the Bargmann group and $A$ a
two-dimensional abelian group, in order to apply the Mackey method to
construct UIRs. As shown in
Appendix~\ref{sec:regul-centr-extend}, the semidirect product
is regular and hence Mackey theory guarantees that all UIRs can be
obtained as induced representations.  There are two kinds of UIRs:
those which are induced from UIRs of the Bargmann group and those
which are induced from UIRs of a Carroll subgroup of the Bargmann
group.  This allows us to borrow (after some translation) the results
of \cite{Figueroa-OFarrill:2024ocf} on the UIRs of the Bargmann group
and of \cite{Figueroa-OFarrill:2023qty} on the UIRs of the Carroll
group in order to determine the UIRs of the centrally extended planon
group.  Rawnsley's theorem \cite{MR387499} guarantees that these UIRs
are obtained by geometric quantisation of coadjoint orbits and we
display this correspondence in Table~\ref{tab:orbits-UIRs}.

We close with a discussion in Section~\ref{sec:discussion}, where we
mention various generalisations and also argue that defects in
crystals at low temperatures are carrollian.

The paper contains three appendices.  In
Appendix~\ref{sec:bargmann-to-planon-summary} we summarise the results
in \cite{Figueroa-OFarrill:2024ocf} about the Bargmann group in planon
language, setting up a useful dictionary between the two.
Appendix~\ref{app:extensions}, contains a technical result about (not
necessarily universal) central extensions.  Finally,
Appendix~\ref{sec:regul-centr-extend} checks that the semidirect
product description of the centrally extended planon group is regular,
a technical result needed to guarantee that the Mackey method is
exhaustive.

\section{Classical planon particles}
\label{sec:planon-particles}

We follow the philosophy of Souriau \cite{MR0260238}, in which
classical planon particles are to be identified with homogeneous
symplectic manifolds of the planon group.  Unlike the case of
fractons, which we treated in this language in
\cite{Figueroa-OFarrill:2023vbj,Figueroa-OFarrill:2023qty}, the planon
group has nontrivial symplectic cohomology and that means that its
homogeneous symplectic manifolds are not necessarily coadjoint orbits
of the planon group, but of a one-dimensional central extension, which
we will describe in Section~\ref{sec:centr-extens-plan}.  After
describing this group we discuss its coadjoint orbits in general
dimension in Section~\ref{sec:coadjoint-orbits}, before specialising
to spatial dimension $n=3$ in Section~\ref{sec:coadjoint-orbits-n=3}.
We list all the coadjoint orbits explicitly as zero loci of equations
and give a representative for each orbit.  These are listed, along
with the stabiliser of the chosen representative, in
Table~\ref{tab:planon-coad-orbs}.  This will then be the departing
point to analyse the elementary particle dynamics and their mobility
in Section~\ref{sec:actions-mobility}.

\subsection{The central extension of the planon group}
\label{sec:centr-extens-plan}

Let $G_{\text{pla}} \cong B \times \RR$ denote the (connected,
simply-connected) planon group, which is isomorphic to the product of
the (connected, simply-connected) Bargmann group $B$ and the additive
group of reals. The Lie algebra $\g_{\text{pla}}$ is the span of
$\left<L_{ab}, P_a, D_a, Z, Q, H\right>$ where the Bargmann algebra
$\b$ in planon language is the span of
$\b = \left< L_{ab}, D_a, P_a, Z, Q\right>$ with $H$ spanning the Lie
algebra of $\RR$. In the language of our recent summary
\cite{Figueroa-OFarrill:2024ocf} of Galilei particles, what we call
$P_a, D_a, Z, Q$ here are there called $B_a, P_a, H, M$, respectively.
We will be reusing many of the results in
\cite{Figueroa-OFarrill:2024ocf}, but translating the generators into
planon language. The precise dictionary is the subject of
Appendix~\ref{sec:bargmann-to-planon-summary}. The Lie brackets of
$\g_{\text{pla}}$ are given by the standard kinematical Lie brackets,
which say that $\r = \left<L_{ab}\right> \cong \so(n)$ is the
rotational subalgebra, $D_a,P_a$ are vectors and $Z,Q,H$ are scalars,
and in addition the following nonzero brackets:\footnote{To highlight
  the translation,
  these are the Bargmann brackets $[B_a, H] = P_a$ and $[B_a, P_b] =
  \delta_{ab} M$.}
\begin{equation}
  \label{eq:planon-brackets-extra}
  [P_a, Z] = D_a  \qquad\text{and}\qquad [P_a, D_b] = \delta_{ab} Q.
\end{equation}
The canonical dual basis for $\g^*_{\text{pla}}$ is $\left<
  \lambda^{ab}, \pi^a, \delta^a, \zeta, \theta, \eta\right>$.
Central extensions are classified up to
equivalence\footnote{Equivalence of central extensions refines the
  notion of Lie algebra isomorphism: two central extensions might be
  isomorphic as Lie algebras and yet inequivalent as central
  extensions, since equivalence requires that the isomorphism be the
  identity both on the Lie algebra being extended and on the central
  ideal.} by the Chevalley--Eilenberg group $H^2(\g_{\text{pla}};\RR)$.
Provided that $n\geq 3$, which we will assume and in fact later
specialise to $n=3$, we can compute this group from the $\r$-basic
subcomplex $C^\bullet(\g_{\text{pla}},\r;\RR)$ consisting of
$\r$-invariant cochains which have no legs along $\r$.  This complex
is very easy to describe:
\begin{equation}
  \label{eq:relative-complex}
  \begin{split}
    C^1(\g_{\text{pla}},\r;\RR) &= \left<\eta, \theta, \zeta\right>\\
    C^2(\g_{\text{pla}},\r;\RR)&= \left<\delta_{ab} \pi^a \wedge
      \delta^b, \eta \wedge \theta, \eta \wedge \zeta,  \theta \wedge \zeta\right>.
  \end{split}
\end{equation}
The Chevalley--Eilenberg differential $\d : C^1(\g_{\text{pla}},\r;\RR)
\to C^2(\g_{\text{pla}},\r;\RR)$ is given by
\begin{equation}
  \label{eq:CE-d}
  \d \eta = \d \zeta = 0 \qquad\text{and}\qquad \d \theta = \delta_{ab} \pi^a \wedge \delta^b.
\end{equation}
We do not need to know what $\d \pi^a$ is for this calculation.  This
is because the only generator in $C^2(\g_{\text{pla}},\r;\RR)$ which
involves $\pi^a$ is obviously a coboundary.  The $2$-cocycles are
spanned by $\left<\delta_{ab} \pi^a \wedge \delta^b, \eta \wedge
  \zeta\right>$ and hence the cohomology is one-dimensional:
\begin{equation}
  \label{eq:2-cohomology}
  H^2(\g_{\text{pla}};\RR) \cong H^2(\g_{\text{pla}},\r;\RR) = \left<[\eta \wedge \zeta]\right>
\end{equation}
spanned by the cohomology class of $\eta \wedge \zeta$.  This signals
the existence of a nontrivial central extension $\g$ of the planon
algebra $\g_{\text{pla}}$ obtained by adding a new generator $W$ and a
new Lie bracket
\begin{equation}
  \label{eq:new-central-bracket}
  [H,Z] = W.
\end{equation}

In summary, we will let $G$ denote the (connected, simply-connected)
centrally extended planon group, with Lie algebra $\g = \left<L_{ab},
  P_a, D_a, Q, Z, H, W\right>$ and Lie brackets
\begin{equation}
  \label{eq:central-extended-planon-brackets}
  [P_a, Z] = D_a, \qquad [P_a, D_b] = \delta_{ab} Q
  \qquad\text{and}\qquad [H,Z] = W,
\end{equation}
beyond the kinematical ones involving $L_{ab}$.  In the nomenclature
of \cite{Morand:2023emw}, this is a $(4,2)$ kinematical Lie algebra
(KLA).

\subsection{Coadjoint orbits}
\label{sec:coadjoint-orbits}

We now work out the coadjoint orbits of the centrally extended planon
group $G$.  We observe that $G \cong B \ltimes A$, where $A$
is the abelian group with Lie algebra $\a = \left<H,W\right>$ and $B$
is the Bargmann group with Lie algebra $\b = \left<L_{ab}, P_a, D_a,
  Q, Z\right>$ as already described above.  In
\cite{Figueroa-OFarrill:2024ocf} we worked out the coadjoint orbits of
the Bargmann group and we can re-use these results thanks to the
following lemma.  The notation $\ann \a$ in the Lemma means the
annihilator of $\a \subset \g$.  This is the subspace of $\g^*$
consisting of all linear functions on $\g$ which vanish identically on
$\a$.  Two of the basic isomorphism theorems in linear algebra are
that $\a^* \cong \g^*/\ann \a$ and dually that $\ann \a \cong (\g/\a)^*$.

\begin{lemma}
  Let $G = B \ltimes A$ where $A$ is abelian and let $\beta \in \ann
  \a \subset \g^*$ be a moment of $B$.  Then for all $a \in A$,
  $\Ad_a^* \beta = \beta$.
\end{lemma}

\begin{proof}
  Let $a \in A$ and $\beta \in \ann \a$.  Then $\Ad_a^*\beta =
  \beta \circ \Ad_{a^{-1}}$.  Let $X \in \g$.

  We claim that $\Ad_{a^{-1}} X = X \mod \a$.  Indeed, using that $A$
  is a normal subgroup
  \begin{align*}
    \exp(t \Ad_{a^{-1}} X) &= a^{-1} \exp(t X) a\\
                          &= a^{-1} \exp(t X) a \exp(-t X) \exp(t X)\\
                          &= a^{-1} a(t) \exp(t X),
  \end{align*}
  where $a(t)$ is a curve in $A$ through $a$.  Differentiating with
  respect to $t$ at $t=0$,
  \begin{equation*}
    \Ad_{a^{-1}} X = X +  \underbrace{a^{-1} \dot a(0)}_{\in \a}.
  \end{equation*}
  Therefore,
  \begin{align*}
    \left<\Ad_a^*\beta, X\right> &= \left<\beta, \Ad_{a^{-1}}X\right>\\
                                  &= \left<\beta, X \mod \a\right> &\tag{by the above claim}\\
                                  &= \left<\beta, X\right>  &\tag{since $\beta \in \ann \a$}
  \end{align*}
  and since this holds for all $X \in \g$, it follows that
  $\Ad_a^*\beta = \beta$.
\end{proof}

This has as a consequence that if $\beta$ is a Bargmann moment, its
coadjoint orbit under $G$ agrees with its Bargmann coadjoint orbit.
Indeed, we can write every $g \in G$ uniquely as $g = b a$ with $b \in
B$ and $a \in A$, and hence, since $\Ad^*$ is a representation,
\begin{equation}
  \Ad_g^* \beta = \Ad_b^* \Ad_a^* \beta = \Ad_b^* \beta,
\end{equation}
where we have used the Lemma in the second equality.

If $\tau = \alpha + \beta \in \g^*$ is such that $\beta \in \ann \a
\cong \b^*$ and $\alpha \in \ann \b \cong \a^*$, then under $g = b a$, we have that
\begin{equation}
  \label{eq:coadjoint-action-abstract}
  \begin{split}
    \Ad_g^* (\alpha + \beta) &= \Ad_b^* \Ad_a^* (\alpha + \beta)\\
    &= \Ad_b^* \beta + \Ad_b^*\Ad_a^*\alpha,
  \end{split}
\end{equation}
where $\Ad_b^* \beta$ can be read off from the results in
\cite{Figueroa-OFarrill:2024ocf}, suitably translated, as in
Appendix~\ref{sec:bargmann-to-planon-summary}.  See in particular
equation~\eqref{eq:barg-coadjoint} in general and
equation~\eqref{eq:barg-coadjoint-3d} for $n=3$.

Let us use these results to calculate the coadjoint action of the
generic element of $G$, denoted
\begin{equation}
  \label{eq:generic-G-element}
  \sg(\varphi,\tilde\varphi,\bbeta,\ba,R,s,w)=
  \sfb(\varphi,\tilde\varphi,\bbeta,\ba,R) \exp(s H + w W),
\end{equation}
on the moment
\begin{equation}
  \label{eq:G-moment}
  \sM(q,\tilde q, \bd,\p,J,E,c) = \tfrac12
  J_{ab}\lambda^{ab} + p_a \pi^a - d_a \delta^a + \tilde q \zeta + q
  \theta + E \eta + c \omega,
\end{equation}
where we have introduced $\eta, \omega \in \g^*$ obeying
$\left<\omega,W\right> =1$ and $\left<\eta, H \right> =1$ and zero
otherwise.  We calculate that $\ad^*_Z \omega = \eta$ and $\ad^*_H
\omega = - \zeta$.  The generic $G$ moment is a sum $\beta + E \eta +
c \omega$, where $\beta$ is a Bargmann moment.

Let $a = \exp(s H + w W)$ and let us calculate
$\Ad_a^* (E \eta + c W)$. Since $W$ is central, $\ad^*_W = 0$ and
hence $\Ad_a^* = \exp(s \ad^*_H)$, resulting in
\begin{equation}
  \Ad_a^* (E \eta + c W) = \exp(s \ad^*_H) (E \eta + c W) = E \eta + c
  \omega - c s \zeta.
\end{equation}
Let $b = \sfb(\varphi,\tilde\varphi,\bbeta, \ba, R)$ and
let us calculate $\Ad_b^*(E\eta + c\omega - cs \zeta)$.  Only $Z$ acts
nontrivially on this subspace, hence
\begin{equation}
  \Ad_b^*(E\eta + c\omega - cs \zeta) = \exp(\tilde\varphi\ad^*_Z)
  (E\eta + c\omega - cs \zeta) = (E + c \tilde\varphi) \eta + c \omega
  - c s \zeta.
\end{equation}

Putting this all together we now have that
\begin{equation}
  \Ad^*_{\sg(\varphi,\tilde\varphi,\bbeta,\ba,R,s,w)}
  \sM(q,\tilde q, \bd,\p,J,E,c) = \sM(q',\tilde q', \bd',\p',J',E',c')
\end{equation}
where, using equation~\eqref{eq:barg-coadjoint},
\begin{equation}
  \label{eq:G-coadjoint-action}
  \begin{split}
    J' &= RJR^T + \bbeta (R \bd + q \boldsymbol{a})^T - (R \bd + q  \boldsymbol{a}) \bbeta^T + (R \p) \ba^T - \ba (R  \p)^T\\
    \p' &= R\p + q \bbeta + \tilde\varphi (R  \bd + q \ba)\\
    \bd' &= R \bd + q \ba\\
    \tilde q' &= \tilde q + R \bd \cdot \ba + \tfrac12 q \|\bd\|^2 - c s\\
    q' &= q\\
    c' &= c\\
    E' &= E + c \tilde\varphi,
  \end{split}
\end{equation}
and in the special case of three dimensions ($n=3$) we have, using
equation~\eqref{eq:barg-coadjoint-3d}, that
\begin{equation}
  \label{eq:G-coadjoint-action-3d}
  \begin{split}
    \bj' &= R\bj - \bbeta \times (R \bd + q \boldsymbol{a}) + \ba \times R \p\\
    \p' &= R\p + q \bbeta + \tilde\varphi (R  \bd + q \ba)\\
    \bd' &= R \bd + q \ba\\
    \tilde q' &= \tilde q + R \bd \cdot \ba + \tfrac12 q
    \|\bd\|^2 -c s\\
    q' &= q\\
    c' &= c\\
    E' &= E + c \tilde\varphi.
  \end{split}
\end{equation}

\subsection{Coadjoint orbits for $n=3$}
\label{sec:coadjoint-orbits-n=3}

Since $Q$ and $W$ span the centre of $\g$, they span the kernel of the
(co)adjoint representation of $\g$.  This means that the group acting
effectively on $\g^*$ is the quotient of $G$ by the centre, which is
the subgroup generated by $Q$ and $W$.  This quotient is isomorphic to
the direct product of the Galilei group with the one-parameter
subgroup generated by $H$. In other words, coadjoint orbits are
homogeneous symplectic manifolds of $\text{Gal}\times \RR$, where
$\text{Gal}$ the group generated by
$\left<\overline{L}_{ab},\overline{P}_a, \overline{D}_a,
  \overline{Z}\right>$, where $X \mapsto \overline{X}$ is the quotient
of $\g$ by the centre.  Homogeneous symplectic manifolds of
$\text{Gal} \times \RR$ are coadjoint orbits of a
\emph{one-dimensional} central extension.  The group
$\text{Gal} \times \RR$ has two-dimensional symplectic cohomology and
the group $G$ is a two-dimensional central extension, which although
not universal (the planon group is not perfect and therefore has no
universal central extension), nevertheless is such that it surjects
onto any one-dimensional central extension.  We can see this at the
level of the Lie algebra and the details are in
Appendix~\ref{app:extensions}.  Hence every one-dimensional central
extension of $\text{Gal} \times \RR$ is the quotient of $G$ by a
one-dimensional central subgroup.  Such a subgroup is generated by a
line in the plane in $\g$ spanned by $Q$ and $W$.  There are two
classes of such lines: those spanned by $W$ and those spanned by $Q -
\lambda W$, with $\lambda\in\RR$.

Quotienting by the one-parameter subgroup generated by $W$ we obtain
the planon group $B \times \RR$, with $B$ the Bargmann group. These
are the coadjoint orbits with $c=0$, which are in one-to-one
correspondence with the Bargmann coadjoint orbits, recalled in
Table~\ref{tab:barg-coad-orbs-planon}, embedded in $\g^*$ by placing
them at a fixed value of $E$.

The orbits with $c\neq 0$ are coadjoint orbits of the quotient of $G$
by the central subgroup generated by $Q - \lambda W$.  Those coadjoint
orbits sit inside the hyperplane of $\g^*$ corresponding to the zeros
of the linear function $Q-\lambda W$, hence those moments with $q =
\lambda c$.  The Lie algebra of this group is generated by
$\left<L_{ab},P_a, D_a, H, Z, W\right>$ and the non-generic brackets are now
\begin{equation}
  [P_a, Z] = D_a, \qquad [P_a, D_b] = \delta_{ab} \lambda W \qquad\text{and}\qquad [H,Z] = W.
\end{equation}
We have two cases to consider: $\lambda = 0$ and $\lambda \neq 0$, in
which case we may set $\lambda = 1$ by rescaling $W$ and $H$.  (Later
in Section~\ref{sec:quant-plan-part}, when we discuss UIRs, it is
more convenient not to do the rescaling and keep $\lambda$ general.)

The case $\lambda = 0$ is when we quotient by the central subgroup
generated by $Q$ and hence corresponds to those coadjoint orbits with
$q=0$ and $c \neq 0$. From equation~\eqref{eq:G-coadjoint-action-3d}
we see that the coadjoint action of $\text{Gal}\times \RR$ is given in
this case by
\begin{equation}
  \label{eq:GalxR-coadjoint-action-lambdaeq0}
  \begin{split}
    \bj' &= R\bj - \bbeta \times R \bd + \ba \times R \p\\
    \p' &= R\p + \tilde\varphi R  \bd\\
    \bd' &= R \bd\\
    \tilde q' &= \tilde q + R \bd \cdot \ba -c s\\
    E' &= E + c \tilde\varphi.
  \end{split}
\end{equation}

In this case it is evident that $\|\bd\|^2$, $\|\p \times
\bd\|^2$ and $E \|\bd\|^2 - c \p \cdot
\bd$ are polynomial functions which are constant on the
orbits.  We may enumerate these orbits as follows.

\begin{itemize}
\item ($\bd=\bzero$)  In this case, the coadjoint action
  simplifies to
  \begin{equation}
    \begin{split}
    \bj' &= R\bj + \ba \times R \p\\
    \p' &= R\p \\
    \tilde q' &= \tilde q - c s\\
    E' &= E + c \tilde\varphi,
    \end{split}
  \end{equation}
  from where we observe that $\|\p\|^2$ and $\p \cdot \bj$ are
  polynomial invariants of the orbits.

  \begin{itemize}
  \item ($\p = \bzero$)  In this case the coadjoint action simplifies
    further to
    \begin{equation}
      \begin{split}
        \bj' &= R\bj \\
        \tilde q' &= \tilde q - c s\\
        E' &= E + c \tilde\varphi,
      \end{split}
    \end{equation}
    and hence $\|\bj\|^2$ is invariant.
    
    If $\bj = \bzero$, then the orbit is the affine plane where
    $E,\tilde\varphi$ can take any values and $c = c_0 \neq 0$.  The
    equations for these two-dimensional orbits are
    \begin{equation}
      \p = \bd = \bj = \bzero, \qquad q=0
      \qquad\text{and}\qquad c=c_0 \neq 0.
    \end{equation}
    We see that there are 11 equations, as expected for two-dimensional
    orbits of 13-dimensional Lie group.
    
    If $\|\bj\|^2 = \ell^2 > 0$, the orbit is the product of a
    sphere of radius $\ell$ and the affine plane above.  The orbits are
    four-dimensional and indeed they are cut out by 9 polynomial
    equations:
    \begin{equation}
      \p = \bd= \bzero, \qquad \|\bj\|^2 = \ell^2, \qquad q=0
      \qquad\text{and}\qquad c= c_0\neq 0.
    \end{equation}

    \item ($\p \neq \bzero$) Let $\|\p\|^2 = p^2>0$ and $\p \cdot \bj =
      hp \in \RR$.  These two equations are supplemented by 5 more
      equations: $\bd=\bzero$, $q=0$ and $c=c_0 \neq 0$.
      The orbit is 6-dimensional and looks like the product of the
      cotangent bundle of the 2-sphere of radius $p$ and the affine
      plane above.  A convenient orbit representative for this orbit
      is $\sM(0,0,\bzero,p \bu, h\bu,0,c_0)$ where $\bu$ is a fixed
      unit vector.  It is not hard to determine that the stabiliser of
      that point is 7-dimensional, as expected.
  \end{itemize}

\item ($\bd \neq \bzero$) Let $\|\bd\|^2 = d^2 > 0$.

  \begin{itemize}
  \item ($\p \times \bd = \bzero$)  This says that $\p$ and
    $\bd$ are collinear: $\p = \pm \tfrac{p}{d} \bd$, with the sign
    denoting whether they are parallel or antiparallel.  In this case,
    $\bj \cdot \bd$ is now invariant and so is $Ed \mp cp$.  The
    coadjoint action becomes now
    \begin{equation}
      \begin{split}
        \bj' &= R\bj + (\pm \tfrac{p}{d} \ba - \bbeta) \times R \bd\\
        \p' &= (\pm \tfrac{p}{d} + \tilde\varphi) R\bd\\
        \bd' &= R \bd\\
        \tilde q' &= \tilde q + R \bd \cdot \ba -c s\\
        E' &= E + c \tilde\varphi.
      \end{split}
    \end{equation}
    The orbits are six-dimensional, being cut out by 7 equations:
    \begin{multline}
      q=0, \quad c=c_0 \neq 0, \quad \p \times \bd =
      \bzero, \quad \|\bd\|^2 = d^2 > 0,\\
      \|E \bd - c \p\|^2 = E_0^2 d^2 \geq 0 \quad\text{and}\quad \bj
      \cdot \bd = hd \in \RR.
    \end{multline}
    We can take as orbit representative the moment $\sM(0,0,d
    \bu, \pm \frac{E_0 d}{c_0} \bu, h \bu,0,c_0)$, where $\bu$ is a
    fixed unit vector.  It is not hard to show that the stabiliser
    subgroup of that moment is 7-dimensional, as expected.  We must
    distinguish a special case of this orbit: when $E_0 = 0$.  In that
    case, we can choose as representative the moment with both $E=0$
    and $\p = \bzero$:  $\sM(0,0,d \bu, \bzero, h \bu,0,c_0)$.  There
    is no longer a sign ambiguity in this case.

  \item ($\p \times \bd \neq \bzero$)  We decompose $\p = \frac{p}{d} \cos\theta
    \bd + \p_\perp$, where $\p_\perp \cdot \bd = 0$.  Let
    $\|\p_\perp\|^2 = p^2 \sin^2 \theta$.  Then the orbit is cut out
    by 5 equations:
    \begin{multline}
      q = 0, \quad c=c_0 \neq 0, \quad \|\bd\|^2 = d^2 >
      0, \quad \|\p \times \bd\|^2 = (p d \sin\theta)^2 > 0\\
      \quad\text{and}\quad E d - c_0 p \cos\theta = \epsilon_0 d \in \RR.
    \end{multline}
    The orbit is parametrised by the values $d,p,\theta,c$, where the
    invariants are $c, d, p\sin\theta$ and $\epsilon_0 := E -
    \frac{cp}{d}\cos\theta$.  We can take as orbit representative the
    moment $\sM(0,0,d \bu, p \sin\theta \bu^\perp -
    \frac{\epsilon_0d}{c_0} \bu,\bzero,0,c_0)$, where $\bu, \bu^\perp$
    are chosen perpendicular unit vectors.  It is not hard to
    determine the stabiliser subgroup and show that it is
    5-dimensional, as expected.
  \end{itemize}
\end{itemize}

The case $\lambda =1$ sets $q=c$ in
equation~\eqref{eq:G-coadjoint-action-3d} to give
\begin{equation}
  \label{eq:GalxR-coadjoint-action-lambdaeq1}
  \begin{split}
    \bj' &= R\bj - \bbeta \times (R \bd + c \boldsymbol{a}) + \ba \times R \p\\
    \p' &= R\p + c \bbeta + \tilde\varphi (R  \bd + c \ba)\\
    \bd' &= R \bd + c \ba\\
    \tilde q' &= \tilde q + R \bd \cdot \ba + \tfrac12 c \|\bd\|^2 -c s\\
    E' &= E + c \tilde\varphi.
  \end{split}
\end{equation}

We can set $s = \tfrac1c (\tilde q + R \bd\cdot \ba) +
\tfrac12 \|\bd\|^2$ and $\tilde\varphi = -E/c$ to set
$\tilde q' = E' = 0$.  Choose $\ba = - \tfrac1c R\bd$ to
set $\bd' = \bzero$ and choose $\bbeta =
-\tfrac1c R \p$ to set $\p' = 0$.  Then we remain with $\bj' = R(\bj -
\tfrac1c \bd\times\p)$.  Therefore we have two kinds of
orbits:
\begin{itemize}
\item those with $\bj' = \bzero$, which are 8-dimensional orbits with
  stabiliser $\SO(3) \times \RR^2$; and
\item those with $\bj' \neq \bzero$, which are 10-dimensional orbits
  with stabiliser the $\SO(2) \times \RR^2$, with $\SO(2)$ the
  subgroup of $\SO(3)$ which preserves $\bj'$.
\end{itemize}

The 8-dimensional orbits are determined by the equation $c\bj -
\bd \times \p= \bzero$, apart from the ones which set the
constant values of $q = c = c_0$.  They consist of the points
\begin{equation}
  \left\{ \sM(c_0,\tilde q, \bd, \p, \tfrac1q
    \bd\times \p,E,c_0)~  \middle | ~\bd, \p \in \RR^3
    \quad\text{and}\quad \tilde q, E \in \RR\right\} \subset \g^*.
\end{equation}
We may take $\sM(c_0,0,\bzero,\bzero,\bzero,0,c_0)$ as orbit
representative.

The 10-dimensional orbits are determined by the equations $\|c\bj -
\bd \times \p\|^2 = \ell^2$, with $\ell > 0$ and $q = c =
c_0$.  They consist of the points
\begin{equation}
  \left\{ \sM(c_0,\tilde q, \bd, \p, \ell \bu,E,c_0)  ~\middle
    | ~   \bd, \p \in \RR^3,\quad \bu \in S^2\subset \RR^3 \quad\text{and}\quad \tilde q, E \in \RR\right\} \subset \g^*.
\end{equation}

\begin{table}[h]
  \centering
    \caption{Coadjoint orbits of the (extended) planon group $G$}
    \label{tab:planon-coad-orbs}
    \setlength{\extrarowheight}{3pt}
  \resizebox{\linewidth}{!}{
    \begin{tabular}{*{3}{>{$}l<{$}}>{$}l<{$}}
      \multicolumn{1}{l}{\#} & \multicolumn{1}{c}{Orbit representative} & \multicolumn{1}{c}{Stabiliser} & \multicolumn{1}{c}{Equations for orbits}\\
                             &
                               \multicolumn{1}{c}{$\alpha=\sM(q,\tilde  q, \bd, \p,\bj,E,c)$} & \multicolumn{1}{c}{$G_\alpha \ni \sg(\varphi,\tilde\varphi,\bbeta,\ba,R,s,w)$} & \\ \midrule \rowcolor{blue!7}
      0& \sM(0,\tilde q_0,\bzero,\bzero,\bzero,E_0,0) & G & c=0, E=E_0, q=0, \tilde q = \tilde q_0, \bd= \p = \bj = \bzero \\
      2& \sM(0,\tilde q_0,\bzero,\bzero,\ell\bu,E_0,0) & \{\sg(\varphi,\tilde\varphi,\bbeta,\ba,R,s,w) \mid R\bu = \bu\} & c=0, E=E_0, q=0, \tilde q = \tilde q_0, \bd = \p = \bzero, \|\bj\|=\ell\\\rowcolor{blue!7}
      2'& \sM(0,0,\bzero,\bzero,\bzero,0,c_0) & \{\sg(\varphi,0,\bbeta,\ba,R,0,w)\} & c=c_0, q=0, \p = \bd = \bj = \bzero\\
      4& \sM(0,\tilde q_0,\bzero, p\bu, h\bu,E_0,0) & \{\sg(\varphi,\tilde\varphi,\bbeta,a \bu,R,s,w) \mid R\bu = \bu\} & c=0, E=E_0, q=0, \tilde q= \tilde q_0, \bd = \bzero, \|\p\|= p, \bj \cdot \p = h p\\\rowcolor{blue!7}
      4'& \sM(0,0,\bzero,\bzero,\ell \bu,0,c_0) & \{\sg{\varphi,0,\bbeta,\ba,R,0,w) \mid R\bu = \bu}\} & c=c_0, q=0, \p = \bd = \bzero, \|\bj\|=\ell\\
      6& \sM(q_0,\tilde q_0,\bzero,\bzero,\bzero,E_0,0) & \{\sg(\varphi,\tilde\varphi,\bzero,\bzero,R,s,w)\} & c=0, E=E_0, q=q_0, \tfrac1{2 q}(\|\bd\|^2- 2 q \tilde q) = \tilde q_0, q\bj= \bd\times\p\\\rowcolor{blue!7}
      6'& \sM(0,0,d \bu,\bzero,\bzero,E_0,0) & \{\sg(\varphi,0,\beta \bu,\ba,R,s,w) \mid R\bu = \bu, \ba \cdot \bu = 0 \} & c=0, E=E_0, q=0, \|\bd\|=d>0, \bd \times \p = \bzero\\
      6''& \sM(0,0,\bzero,p \bu, h \bu,0,c_0) & \{\sg(\varphi,0,\bbeta,a \bu, R,0,w) \mid R\bu = \bu\} & c=c_0, q=0, \bd = \bzero, \|\p\|=p, \bj \cdot \p = hp\\\rowcolor{blue!7}
      6_\pm'''& \sM(0,0,d\bu,\pm \frac{E_0 d}{c_0}\bu,h\bu,0,c_0) & \{\sg(\varphi,0,\bbeta,\ba,R,0,w) \mid R\bu =\bu, (d\bbeta \mp p\ba) \times \bu = \bzero\} & c=c_0, q=0,\p \times \bd = \bzero, \|\bd\|=d, \bj \cdot \bd = h d,\|E\bd-c\p\| = E_0 d\\
      6_0'''& \sM(0,0,d\bu,\bzero,h\bu,0,c_0) & \{\sg(\varphi,0,\bbeta,\ba,R,0,w) \mid R\bu =\bu, d\bbeta \times \bu = \bzero\} & c=c_0, q=0,\p =\frac{E}{c} \bd, \|\bd\|=d, \bj \cdot \bd = h d\\
      8& \sM(q_0,\tilde q_0,\bzero,\bzero, \ell \bu,E_0,0) & \{\sg(\varphi,\tilde\varphi,\bzero,\bzero,R,s,w) \mid R\bu = \bu\} & c=0, E=E_0, q=q_0, \tfrac1{2 q}(\|\db\|^2- 2 q\tilde q) = \tilde q_0, \|q\bj - \bd \times\p \| = \ell\\\rowcolor{blue!7}
      8'& \sM(0,0,d\bu, p\bu^\perp,\bzero,E_0,0) & \{\sg(\varphi,0,\beta \bu, a \bu^\perp,\pm I,s,w)\} & c=0, E=E_0, q = 0, \|\bd\| = d, \|\bd \times \p\| = d p \\
      8''& \sM(0,0,d\bu, p \sin\theta \bu^\perp - \frac{\epsilon_0 d}{c_0} \bu, \bzero,0,c_0) & \{\sg(\varphi,0,\bbeta,\ba,\pm I,0,w) \mid \ba \times \p = \bbeta \times \bd\} & c=c_0, q=0, \|\bd\|= d , \|\p \times \bd\| = pd\sin\theta, (E-\epsilon_0) d^2 = c \p \cdot \bd\\\rowcolor{blue!7}
      8'''& \sM(c_0,0,\bzero,\bzero,\bzero,0,c_0) & \{\sg(\varphi,0,\bzero,\bzero,R,0,w)\} & c=q=c_0, c \bj = \bd \times \p\\
      10& \sM(c_0,0,\bzero,\bzero,\ell \bu,0,c_0) & \{\sg(\varphi,0,\bzero,\bzero,R,0,w) \mid R\bu = \bu\} & c=q=c_0, \|\bj - \frac1{c}\bd \times \p\| = \ell\\
   \bottomrule
    \end{tabular}
  }
  \caption*{This table lists the different
    coadjoint orbits of the (extended) planon group ordered by
    increasing dimension.  In each case we exhibit a label (from which
    one can read off the dimension), an orbit representative
    $\alpha\in \g^*$, its stabiliser subgroup $G_\alpha$ inside the
    planon group and the equations defining the orbit.  In this table,
    $\bu$ and $\bu^\perp$ stand for perpendicular unit vectors in
    $\RR^ 3$.  Wherever they appear, the parameters $c_0$ and $q_0$
    are nonzero, whereas the parameters $\ell$,
    $d$, $p$ and $E_0$ are positive.  The parameters $h$ and
    $\epsilon_0$ are arbitrary real numbers. The $\pm I$ in the
    stabiliser is due to the fact that we are working with the
    simply-connected spin group and $\pm I$ is the centre, which is in
    the kernel of the geometric action of the spin group on vectors
    via rotations.}
\end{table}

\section{Elementary particles and mobility}
\label{sec:actions-mobility}

In this section we study the dynamics of the elementary particles
associated with the different coadjoint orbits, paying particular
attention to their (restricted) mobility. Coadjoint orbits of the
centrally extended planon group are homogeneous symplectic manifolds
of the planon group.  Let $\eO_\alpha$ be the coadjoint orbit
corresponding to $\alpha \in \g^*$.  The orbit map $\pi_\alpha : G \to
\eO_\alpha$ sends every $g \in G$ to $\Ad_g^* \alpha$.  Let $\omega
\in \Omega^2(\eO_\alpha)$ denote the Kirillov--Kostant--Souriau
symplectic form on the coadjoint orbit: it is a closed non-degenerate
$2$-form.  Pulling back $\omega$ via the orbit map, we get a
presymplectic form $\pi_\alpha^*\omega \in \Omega^2(G)$ in the group:
a closed $2$-form.  In fact, a calculation shows that it is not just
closed, but actually exact: $\pi_\alpha^*\omega =
-d\left<\alpha,\vartheta\right>$, where $\vartheta \in \Omega^1(G;\g)$
is the left-invariant Maurer--Cartan one-form on $G$ and
$\left<-,-\right>$ denotes the dual pairing.  The primitive
$\left<\alpha,\vartheta\right>$ defines a variational problem for
curves on $G$.  If $\gamma : I \to G$, where $I$ is some interval
containing $0$ in the real line, then we define
\begin{equation}\label{eq:action-functional-general}
  S[\gamma] := \int_I \left<\alpha, \gamma^*\vartheta\right> = \int_I
  \left<\alpha, \gamma(\tau)^{-1} \dot \gamma(\tau)\right> d\tau,
\end{equation}
where $\tau \in I$ is the parameter along the curve.  In
Section~\ref{sec:MC-one-form} we calculate the left-invariant
Maurer--Cartan one-form on $G$ and in
Section~\ref{sec:particle-dynamics} we perform a preliminary geometric
analysis of the particle trajectories resulting from the above
variational principle.  In
Section~\ref{sec:particle-dynamics-orbitwise} we concentrate on some
of the more interesting coadjoint orbits (those which are coadjoint
orbits of the planon group itself) and discuss action functionals for
trajectories on the corresponding aristotelian spacetime where the
planons live.

\subsection{Maurer--Cartan one-form}
\label{sec:MC-one-form}

We record here the pull-back to the space of group parameters of the
left-invariant Maurer--Cartan one-form on $G$.  Let $g = b a$, with $b
\in B$ and $a = \exp(t H + w W)$.  Then
\begin{equation}
  g^{-1} dg = a^{-1} (b^{-1}db) a+ a^{-1} d a = \exp(t \ad_H)
  (b^{-1}db) + dt H + dw W,
\end{equation}
where we have used that $W$ is central.  We can read off $b^{-1}db$
from equation~\eqref{eq:MC-one-form-Bargmann}:
\begin{equation}
  b^{-1}db = \sA(d\varphi - \ba^T d\bbeta - \tfrac12
  \|\ba\|^2d\tilde\varphi, d\tilde\varphi, R^T(d\bbeta +
  \ba d\tilde\varphi), R^T d\ba, R^T dR)
\end{equation}
and hence
\begin{equation}
  \exp(t \ad_H) (b^{-1}db)  = \sA(d\varphi - \ba^T d\bbeta - \tfrac12
  \|\ba\|^2d\tilde\varphi, d\tilde\varphi, R^T(d\bbeta +
  \ba d\tilde\varphi), R^T d\ba, R^T dR) + t d\tilde\varphi W,
\end{equation}
which we can put together to arrive at an expression for $g^{-1}dg$.
Contracting that expression with the generic moment $\sM(q,
\tilde q, \bd,\p, J,E,c)$ we find (using
equation~\eqref{eq:MC-one-form-contracted-Bargmann}),
\begin{multline}
  \label{eq:MC-one-form-contracted}
  \left<\sM(q,\tilde q, \bd,\p, J,E,c), g^{-1}dg\right> =
  c (dw + t d \tilde\varphi) + E dt + q d\varphi - (\tilde q +
  \tfrac12 q \|\ba\|^2 + (R\bbeta)^T\ba) d\tilde\varphi\\ -
  (R\bd + q \ba)^T d\bbeta + (R\p)^T d\ba +
  \tfrac12 \Tr J^T R^T d R.
\end{multline}
Specialising to three dimensions, we find
\begin{multline}
  \label{eq:MC-one-form-contracted-3d}
    \left<\sM(q,\tilde q, \bd,\p, \bj,E,c), g^{-1}dg\right> =
  c (dw + t d \tilde\varphi) + E dt + q d\varphi - (\tilde q +
  \tfrac12 q \|\ba\|^2 + R\bbeta\cdot \ba) d\tilde\varphi\\ -
  (R\bd + q \ba)\cdot d\bbeta + R\p\cdot d\ba +
  \bj \cdot \varepsilon^{-1}(R^{-1}dR).
\end{multline}

\subsection{Geometric particle dynamics}
\label{sec:particle-dynamics}

We interpret the planon particle dynamics as taking place in an
aristotelian spacetime whose momenta lie in the different coadjoint
orbits of $G$. The discussion here follows that of
\cite[Section~9]{Figueroa-OFarrill:2024ocf} which itself follows the
discussion in \cite[Appendix~A.4]{Figueroa-OFarrill:2023vbj} (see also
e.g.,~\cite{Oblak:2016eij,Barnich:2022bni,Basile:2023vyg,Lahlali:2024qnk,Barnich:2025ton}
and references therein). As shown in that latter reference, the curves
$\gamma(\tau)$ extremising the action functional in
equation~\eqref{eq:action-functional-general} are of the form
$\gamma(\tau) = g_0 c(\tau)$, where $g_0 \in G$ and
$c: I \to G_\alpha$ is any curve in the stabiliser of $\alpha$. Under
the orbit map, $\pi_\alpha(\gamma(\tau)) = \Ad_{g_0}^*\alpha$, which
is a fixed point in the coadjoint orbit $\eO_\alpha$ of $\alpha$. The
curve $\gamma$ can now be mapped to any other homogeneous space of $G$
via the corresponding orbit map. For example, let $M$ denote the
aristotelian spacetime with (non-effective) Klein pair $(\g, \g_o)$,
where $\g_o= \left<L_a, D_a, Q, W, Z\right>$. Let $G_o$ be the
connected subgroup of $G$ with Lie algebra $\g_o$ and let $o \in M$ be
any point with stabiliser $G_o$. We may use the orbit map
$\pi_o: G \to M$ to map the curves $\gamma(\tau)$ on $G$ to curves on
$M$: these are the classical trajectories whose momenta lie in
$\eO_\alpha$.

If the curve $c$ landed in the intersection $G_\alpha \cap G_o$ then
clearly this would correspond to a trajectory which is fixed at the
point $g_0 \cdot o \in M$. Hence we can choose a complement $\fm$ to
$\g_\alpha \cap \g_o$ in $\g_\alpha$, so that
$\g_\alpha = \left(\g_\alpha \cap \g_o \right) \oplus \fm$ and give
local coordinates to $G_\alpha$ by choosing bases for
$\g_\alpha \cap \g_o$ and for $\fm$ and write elements in $G_\alpha$
locally as a product of exponentials $\exp(X)\exp(Y)$ with
$Y \in \g_\alpha \cap \g_o$ and $X \in \fm$. The local coordinates are
then the coefficients of $X$ and $Y$ relative to the chosen bases. The
complement $\fm$ detects the possible directions along which the
particle trajectory can evolve.

In Table~\ref{tab:stabilisers} we list the stabiliser algebras
$\g_\alpha$ for each type of coadjoint orbit in
Table~\ref{tab:planon-coad-orbs}, the intersection
$\g_\alpha \cap \g_o$ and a choice of vector space complement $\fm$
such that $\g_\alpha = \left(\g_\alpha \cap \g_o \right) \oplus \fm$.

\begin{table}[h]
  \centering
    \caption{Stabilisers associated to particle dynamics}
    \label{tab:stabilisers}
    \setlength{\extrarowheight}{3pt}
    \resizebox{\linewidth}{!}{
    \begin{tabular}{*{5}{>{$}l<{$}}}
      \# & \alpha =\sM(q,\tilde q, \bd,\p,\bj,E,c) \in \g^* & \g_\alpha & \g_\alpha \cap \g_o& \fm\\
      \midrule\rowcolor{blue!7}
      0& \sM(0,\tilde q_0,\bzero,\bzero,\bzero,E_0,0) & \g & \g_o & \left<H,P_a\right>\\
      2& \sM(0,\tilde q_0,\bzero,\bzero,\ell\bu,E_0,0) & \left<W,H,Q,Z,D_a,P_a,L_3\right> & \left<W,Q,Z,D_a,L_3\right> & \left<H,P_a\right>\\\rowcolor{blue!7}
      2'& \sM(0,0,\bzero,\bzero,\bzero,0,c_0) & \left<W,Q,D_a,P_a,L_a\right> & \left<W,Q,D_a,L_a\right> & \left<P_a\right>\\
      4& \sM(0,\tilde q_0,\bzero, p\bu, h\bu,E_0,0) & \left<W,H,Q,Z,D_a,P_3,L_3\right> & \left<W,Q,Z,D_a,L_3\right> & \left<H,P_3\right>\\\rowcolor{blue!7}
      4'& \sM(0,0,\bzero,\bzero,\ell \bu,0,c_0) & \left<W,Q,D_a,P_a,L_3\right> & \left<W,Q,D_a,L_3\right> & \left<P_a\right>\\
      6& \sM(q_0,\tilde q_0,\bzero,\bzero,\bzero,E_0,0) & \left<W,H,Q,Z,L_a\right> & \left<W,Q,Z,L_a\right> & \left<H\right>\\\rowcolor{blue!7}
      6'& \sM(0,0,d \bu,\bzero,\bzero,E_0,0) & \left<W,H,Q,D_3,P_1,P_2,L_3\right> & \left<W,Q,D_3,L_3\right> & \left<H,P_1,P_2\right>\\
      6''& \sM(0,0,\bzero,p \bu, h \bu,0,c_0) & \left<W,Q,D_a,P_3,L_3\right> & \left<W,Q,D_a,L_3\right> & \left<P_3\right>\\\rowcolor{blue!7}
      6_\pm'''& \sM(0,0,d\bu,\pm \frac{E_0 d}{c_0}\bu,h\bu,0,c_0) & \left<W,Q,D_3,L_3,P_3,P_1\mp\frac{E_0 d}{c_0}D_1,P_2\pm\frac{E_0 d}{c_0}D_2\right> & \left<W,Q,D_3,L_3\right> & \left<P_1\mp\frac{E_0 d}{c_0}D_1, P_2\mp\frac{E_0 d}{c_0}D_2,P_3 \right>\\
      6_0'''& \sM(0,0,d\bu,\bzero,h\bu,0,c_0) & \left<W,Q,D_3,L_3,P_3,P_1,P_2\right> & \left<W,Q,D_3,L_3\right> & \left<P_a\right>\\\rowcolor{blue!7}
      8& \sM(q_0,\tilde q_0,\bzero,\bzero, \ell \bu,E_0,0) & \left<W,H,Q,Z,L_3\right> & \left<W,Q,Z,L_3\right> & \left<H\right>\\
      8'& \sM(0,0,d\bu, p\bu^\perp,\bzero,E_0,0) & \left<W,H,Q,D_3,P_1\right> & \left<W,Q,D_3\right> & \left<H,P_1\right>\\\rowcolor{blue!7}
      8''& \sM(0,0,d\be_3, p \sin\theta \be_2 - \frac{\epsilon_0 d}{c_0} \be_3, \bzero,0,c_0) & \left<W,Q,D_3, P_1-\frac{\epsilon_0}{c_0}D_1, P_3 - \frac{p}{d}\sin\theta D_1\right> & \left<W,Q,D_3\right> & \left<P_1-\frac{\epsilon_0}{c_0}D_1, P_3 - \frac{p}{d}\sin\theta D_1\right>\\
      8'''& \sM(c_0,0,\bzero,\bzero,\bzero,0,c_0) & \left<W,Q,L_a\right> & \left<W,Q,L_a\right> & 0\\\rowcolor{blue!7}
      10& \sM(c_0,0,\bzero,\bzero,\ell \bu,0,c_0) & \left<W,Q,L_3\right> & \left<W,Q,L_3\right> & 0\\
   \bottomrule
    \end{tabular}
  }
\end{table}

One can read off from this table, that particles corresponding to
orbits with $c = c_0 \neq 0$ do not evolve in time: as the time
translation generator $H$ is not in the stabiliser of $\alpha$: this
just reflects that energy is not conserved for such orbits and is
indeed unbounded from both below and above.  The remaining orbits,
labelled $0,2,4,6,6',8$ and $8'$, are studied in turn below.

We can read off some properties of those orbits from the table:
particles with momenta in coadjoint orbits of types $0,2$ are
unrestricted; those in orbits of types $4,8'$ move along ``lines'';
those in orbits of type $6'$ move along ``planes''; whereas those in
orbits of types $6,8$ do not move.  These statements do not
necessarily mean that mobility is restricted, but they simply describe
the existence of such particle trajectories.  This is analogous to how
a massive (galilean or minkowskian) particle can be boosted to its
rest frame, where it appears as if it does not move.

\subsection{Particle dynamics}
\label{sec:particle-dynamics-orbitwise}

Based on our understanding of the coadjoint orbits, we will now derive
actions for the orbits where $c = 0$, and we shall analyse their
mobility restrictions. Since with the standard planon interpretation
of the generators a nonzero $c$ results in an energy that is unbounded
from below, we will exclude such cases from the current analysis.

Since $c=0$ the coadjoint orbits and actions can again be mapped to
the Bargmann case (see, e.g.,~\cite{Figueroa-OFarrill:2024ocf}),
although with different physical interpretation. In particular, some
of the orbits most relevant in the context of planons, such as those
describing dipoles, correspond to exotic massless orbits in the
Bargmann framework.

In the following, to simplify the interpretation of the orbits, we
will replace $\boldsymbol{a}$ with $\boldsymbol{x}$, where
$\boldsymbol{x}$ represents the position of the corresponding particle
and the orbit numbers and representatives are taken from
Table~\ref{tab:planon-coad-orbs}.

\subsubsection{Orbit \#0 (trace particle)}
\label{sec:orbit-0-trace}

Let us consider a coadjoint orbit with a representative given by
\begin{equation}
  \label{eq:rep_trace_part}
  \alpha =\sM(0,\tilde{q}_{0},\boldsymbol{0},\boldsymbol{0},\boldsymbol{0},E_0,0) \, ,
\end{equation}
where $E_{0}$ and $\tilde{q}_{0}$ denote the energy and the trace
charge, respectively.

The Lagrangian is explicitly obtained from the Maurer-Cartan
form~\eqref{eq:MC-one-form-contracted-3d}, resulting in
\begin{equation}
  L\left[\tilde{\varphi},t\right]=\tilde{q}_{0}\dot{\tilde{\varphi}}-E_{0}\dot t \, .
\end{equation}
Since this Lagrangian is a boundary term, there is no dynamic
associated with it. Given that the stabiliser is the entire group,
this ``trace particle'' can also be interpreted as a vacuum
configuration. Indeed, from the perspective of the Bargmann group,
this orbit is viewed as a vacuum (see for example,
\cite{Figueroa-OFarrill:2024ocf}). It is important to note that the
symmetries impose no restrictions on the ``mobility.''

\subsubsection{Orbit \#2 (spinning trace particle)}
\label{sec:orbit-2-spinning}

Let us consider the following representative:
\begin{equation}
  \alpha =\sM(0,\tilde{q}_{0},\boldsymbol{0},\boldsymbol{0},\boldsymbol{j},E_0,0)
\end{equation}
where $\boldsymbol{\left\Vert j\right\Vert }=\ell$. The Lagrangian
associated with this coadjoint orbit is given by
\begin{equation}
  L\left[\tilde{\varphi},t,R\left(\boldsymbol{\phi}\right)\right]=\tilde{q}_{0}\dot{\tilde{\varphi}}-E_{0}\dot{t}+\bj \cdot \varepsilon^{-1}(R^{-1}dR) \, .
\end{equation}
Note that only the spin part of the Lagrangian is not a boundary
term. Consequently, this orbit possesses only spin degrees of freedom.

Following~\cite{Figueroa-OFarrill:2023vbj}, it is useful to adopt the
following parametrisation for the rotation matrix:
\begin{equation}
  R\left(\boldsymbol{\phi}\right)=e^{\phi_{1}\varepsilon_{1}}e^{\phi_{2}\varepsilon_{2}}e^{\phi_{3}\varepsilon_{3}},\label{eq:Rot_matrix}
\end{equation}
with $\left(\varepsilon_{a}\right)_{bc}=-\varepsilon_{abc}$. In the
particular case where the angular momentum is aligned with the
$z$-axis, we can write:
\begin{equation}
  L\left[\tilde{\varphi},t,\boldsymbol{\phi}\right]=\tilde{q}_{0}\dot{\tilde{\varphi}}-E_{0}\dot{t}+\ell\left(\dot{\phi}_{3}+\sin\left(\phi_{2}\right)\dot{\phi}_{1}\right) \, .
\end{equation}
If one removes the boundary terms, the Lagrangian can be rewritten as
\begin{equation}
  L\left[\Pi^{1},\phi_{1}\right]=\Pi^{1}\dot{\phi}_{1},
\end{equation}
where $\Pi^{1}:=\ell\sin\phi_{2}$ is the canonical conjugate to
$\phi_{1}$.

In the context of the Bargmann group interpretation, this
configuration represents a ``spinning vacuum.''

\subsubsection{Orbit \#4 (spinning trace particle with momentum)}
\label{sec:orbit-4-spinning}

For this orbit, it is necessary to consider the following representative:
\begin{align}
  \alpha=\sM\left(0,\tilde{q}_{0},\boldsymbol{0},\boldsymbol{p},\boldsymbol{j},E_0,0\right),
\end{align}
where $\left\Vert \boldsymbol{p}\right\Vert ^{2}=p_{0}^{2}$ is a constant, with
$p_{0}>0$. Additionally, $\boldsymbol{j}\cdot\boldsymbol{p}=p_{0}h_{0}$
for a real constant $h_{0}$. As a consequence, we can express
\begin{align}
  \boldsymbol{p}=p_{0}\hat{\boldsymbol{n}},
\end{align}
where 
\begin{align}
  \hat{\boldsymbol{n}}=\left(\sin\phi_{2},-\sin\phi_{1}\cos\phi_{2},-\cos\phi_{1}\cos\phi_{2}\right).
\end{align}
The Lagrangian corresponding to this orbit is then given by
\begin{align}
  L\left[\tilde{\varphi},t,\boldsymbol{x},\boldsymbol{\phi}\right]=\tilde{q}_{0}\dot{\tilde{\varphi}}+\boldsymbol{p}\cdot\dot{\boldsymbol{x}}-E_{0}\dot{t}+\bj \cdot \varepsilon^{-1}(R^{-1}dR).
\end{align}
Assuming that $\boldsymbol{j}$ and $\boldsymbol{p}$ are both aligned along the $z$-axis, such that $\boldsymbol{j}=\left(0,0,j_{0}\right)$ and $\boldsymbol{p}=\left(0,0,p_{0}\right)$,
the Lagrangian can be expressed as 
\begin{align}
L\left[\tilde{\varphi},t,\boldsymbol{x},\boldsymbol{\phi}\right]=\tilde{q}_{0}\dot{\tilde{\varphi}}+h_{0}\left(\dot{\phi}_{3}+\sin\left(\phi_{2}\right)\dot{\phi}_{1}\right)+p_{0}\hat{\boldsymbol{n}}\cdot\dot{\boldsymbol{x}}-E_{0}\dot{t}.
\end{align}
If the boundary terms are discarded, and the following canonical momenta are introduced:
\begin{align}
  \boldsymbol{\pi}:=\frac{\partial L}{\partial\dot{\boldsymbol{x}}}=p_{0}\hat{\boldsymbol{n}}\qquad\qquad\qquad\Pi^{1}:=\frac{\partial L}{\partial\dot{\phi}_{1}}=h_{0}\sin\phi_{2},\label{eq:momenta_with_spin}
\end{align}
which satisfy the constraints
\begin{align}
  \left\Vert \boldsymbol{\pi}\right\Vert ^{2}-p_{0}^{2}=0\qquad\qquad\qquad p_{0}\Pi^{1}-h_{0}\pi_{1}=0,
\end{align}
then the Lagrangian in canonical form can be written as
\begin{equation}
L_{\text{can}}\left[\boldsymbol{x},\boldsymbol{\pi},\phi_{1},\Pi^{1},\eta,\eta_{1}\right] =\Pi^{1}\dot{\phi}_{1}+\boldsymbol{\pi}\cdot\dot{\boldsymbol{x}}-\eta\left(\left\Vert \boldsymbol{\pi}\right\Vert ^{2}-p_{0}^{2}\right)-\eta_{1}\left(p_{0}\Pi^{1}-h_{0}\pi_{1}\right). 
\end{equation}

The dynamics of this particle can be understood in terms of
carrollian, and galilean physics. In the carrollian context, the
dynamics arising from this Lagrangian were explored in Section 3.4 of
\cite{Figueroa-OFarrill:2023vbj}, where it was referred to as
``massless carrollion.'' This Lagrangian corresponds to one of the
massless coadjoint orbits of the Carroll group. The specific case
where the spin vanishes had been previously studied
in~\cite{deBoer:2021jej}. On the other hand, this Lagrangian appears
in one of the massless representations of the Bargmann group, where
the position of the particle is replaced by its velocity. This case
was analysed in Section~8.3.3 of~\cite{Figueroa-OFarrill:2024ocf}.

\subsubsection{Orbit \#6 (monopole)}
\label{sec:orbit-6-monopole}

The coadjoint orbit in this case is determined by the following
representative:
\begin{equation}
  \alpha=\sM\left(q_{0},\tilde{q}_{0},\boldsymbol{0},\boldsymbol{0},\boldsymbol{0},E_0,0\right) .
\end{equation}
The Lagrangian is given by
\begin{equation}
  \label{eq:monopole_new}
  L\left[\varphi,\tilde{\varphi},\boldsymbol{\beta},\boldsymbol{x},t\right]=-E_{0}\dot{t}+q_{0}\dot{\varphi}+\tilde{q}_{0}\dot{\tilde{\varphi}}-\frac{1}{2} q_{0}\left\Vert \boldsymbol{x}\right\Vert ^{2}\dot{\tilde{\varphi}}- q_{0}\boldsymbol{x}\cdot\dot{\boldsymbol{\beta}} \, .
\end{equation}
We fix the gauge $\tau = t$ and ignore the remaining boundary
terms. Furthermore, it is useful to perform integration by parts on the
last two terms. The Lagrangian can then be expressed as follows, where
the derivatives are now with respect to the physical time $t$
\begin{equation}
  \label{eq:mono-first-step}
  L\left[\tilde{\varphi},\boldsymbol{\beta},\boldsymbol{x}\right]=q_{0}\left(\boldsymbol{x}\tilde{\varphi}+\boldsymbol{\beta}\right)\cdot\dot{\boldsymbol{x}}.
\end{equation}
The Lagrangian can be written in Hamiltonian form by introducing the
conjugate variable to $\boldsymbol{x}$, defined as
\begin{equation}
  \label{eq:mom_mon}
\boldsymbol{\pi}:=\frac{\pd L}{\pd \dot{\bm{x}}}=q_{0}\left(\boldsymbol{x}\tilde{\varphi}+\boldsymbol{\beta}\right) .
\end{equation}
The Lagrangian then becomes
\begin{equation}
L\left[\boldsymbol{x},\boldsymbol{\pi}\right]=\boldsymbol{\pi}\cdot\dot{\boldsymbol{x}}.\label{eq:monopole_Lag}
\end{equation}
The equations of motion are given by 
\begin{align}
  \label{eq:eom-monop}
  \dot{\boldsymbol{x}}&=\bm{0} & \dot{\boldsymbol{\pi}}&=\bm{0} \,
\end{align}
showing that this particle cannot move. Although we refer to it as a monopole, it actually carries both electric and quadrupole charges.

It is interesting to contrast the monopole with the corresponding
Bargmann case, which is the well-known massive Galilei particle (see
Appendix~\ref{sec:bargmann-to-planon-summary} for more details
concerning the dictionary).
The starting point is the action
\begin{align}
  L[a_{+},t,\bm{x},\bm{v}] =0+ m \dot a_{+} - E_{0}\dot t  - \frac{1}{2}m\left\Vert\bm{v}\right\Vert^{2} \dot t + m \bm{v}\cdot \dot{\bm{x}} \, ,
\end{align}
which agrees with \eqref{eq:monopole_new}, up to the first term which
has no counterpart in the Galilei case, and immaterial changes of
signs. Again we gauge fix the physical time $\tau=t$ and ignore the
remaining boundary terms, which leads to
\begin{align}
  \label{eq:gal-mid}
  L[\bm{x},\bm{v}] = - \frac{1}{2}m\left\Vert\bm{v}\right\Vert^{2}  + m \bm{v}\cdot \dot{\bm{x}} \, ,
\end{align}
that must be contrasted with~\eqref{eq:mono-first-step}. In
particular, varying \eqref{eq:gal-mid} with respect to $\bm{x}$ leads
to
\begin{equation}
  \dot{\boldsymbol{v}}= \boldsymbol{0} \, ,
\end{equation}
which is the galilean analog of $\dot{\boldsymbol{x}}=\boldsymbol{0}$
of \eqref{eq:eom-monop}.

We can now vary \eqref{eq:gal-mid} with respect to $\bm{v}$, which leads to the equation
$\bm{v}=\dot{\bm{x}}$. We are then allowed to substitute this back
into the action to obtain
\begin{align}
  L[\bm{x}] =  \frac{1}{2}m\left\Vert\dot{\bm{x}}\right\Vert^{2}  \, .
\end{align}
As is well known, these particles move along straight lines. 

This example shows that the physical interpretation of the parameters
do play an important role for the mobility the particle models. 

\subsubsection{Orbit \#8 (spinning monopole)}
\label{sec:orbit-8-spinning}

Let us consider the following representative of the coadjoint orbit
\begin{equation}
  \alpha=\sM\left(q_{0},\tilde{q}_{0},0,0,\boldsymbol{j},E_0,0\right),
\end{equation}
where $\boldsymbol{\left\Vert j\right\Vert }^{2}=\ell^{2}$ is a constant.

The Lagrangian is given by
\begin{equation}
  L\left[\varphi,\tilde{\varphi},\boldsymbol{x},\boldsymbol{\beta},t,R\left(\boldsymbol{\phi}\right)\right]
  =q_{0}\dot{\varphi}+\left(\tilde{q}_{0}-\frac{1}{2}q_{0}\left\Vert \boldsymbol{x}\right\Vert ^{2}\right)\dot{\tilde{\varphi}}-q_0 \boldsymbol{x}\cdot\dot{\boldsymbol{\beta}}-E_0\dot{t}+\bj \cdot \varepsilon^{-1}(R^{-1}dR).
\end{equation}
Using the parametrisation for the rotation matrix in
\eqref{eq:Rot_matrix}, the spin part becomes
\begin{equation}
L_{\text{spin}}\left(\boldsymbol{\phi}\right)=\ell\left(\dot{\phi}_{3}+\sin\phi_{2}\right)\dot{\phi}_{1}.\label{eq:Lspin}
\end{equation}

By introducing the canonical momenta as defined in Eq. \eqref{eq:momenta_with_spin} and neglecting boundary terms, the Lagrangian in canonical form can be expressed as:
\begin{equation}
    L\left[\boldsymbol{x},\boldsymbol{\pi},\phi_{1},\Pi^{1}\right]=\Pi^{1}\dot{\phi}_{1}+\boldsymbol{\pi}\cdot\dot{\boldsymbol{x}}.
\end{equation}
This action depends on 8 independent canonical variables, which
coincides with the dimension of the coadjoint orbit.  In the context
of the Bargmann interpretation, it describes a massive spinning
particle.

\subsubsection{Orbit \#6' (dipole)}
\label{sec:orbit-6-dipole}

Let us consider the following representative:
\begin{equation}
  \alpha=\sM\left(0,0,\boldsymbol{d},\boldsymbol{0},\boldsymbol{0},E_0,0\right),
\end{equation}
where $\left\Vert \boldsymbol{d}\right\Vert ^{2}=d_{0}^{2}$ with $d_0$ being a constant.

The Lagrangian is given by
\begin{equation}
  L\left[\boldsymbol{\phi},\boldsymbol{x},\boldsymbol{\beta},\tilde{\varphi},t\right]
  =\left(R\boldsymbol{d}\right)\cdot\boldsymbol{x}\dot{\tilde{\varphi}}+\left(R\boldsymbol{d}\right)\cdot\dot{\boldsymbol{\beta}}-E_{0}\dot{t} \, .
\end{equation}
It is useful to define the following fields:
\begin{equation}
  \boldsymbol{D} := R\boldsymbol{d}_0 \qquad\qquad \tilde{Q} := \boldsymbol{D} \cdot \boldsymbol{x}\;,
\end{equation}
where the field \(\boldsymbol{D}\) represents the dipole moment
obtained by applying a generic rotation to the representative dipole
moment \(\boldsymbol{d}\). The field \(\tilde{Q}\) can be interpreted
as describing the degree of freedom associated with the trace
charge. They satisfy the following constraints
\begin{equation}
  \left\Vert \boldsymbol{D}\right\Vert ^{2}-d_{0}^{2} \approx 0\qquad\qquad\tilde{Q}-\boldsymbol{D}\cdot\boldsymbol{x}\approx 0 \, .
\end{equation}
Thus, we can write
\begin{align}
L\left[t,E,\tilde{\varphi},\boldsymbol{D},\boldsymbol{\beta},\boldsymbol{x},\boldsymbol{\pi},\eta_{1},\eta_{2},\boldsymbol{\eta}_3,N\right]=&-E\dot{t}+\tilde{Q}\dot{\tilde{\varphi}}+\boldsymbol{D}\cdot\dot{\boldsymbol{\beta}}+\boldsymbol{\pi}\cdot\dot{\boldsymbol{x}} \nonumber \\
&-\eta_{1}\mathcal{C}_{1}-\eta_{2}\mathcal{C}_{2}-\boldsymbol{\eta}_{3}\cdot\boldsymbol{\mathcal C}_3 - N \mathcal C_4\, ,
\end{align}
with the constraints
\begin{align}
  \mathcal{C}_{1}&:=\left\Vert \boldsymbol{D}\right\Vert^{2}-d_{0}^{2}\approx 0\\
  \mathcal{C}_{2}&:=\tilde{Q}-\boldsymbol{D}\cdot\boldsymbol{x}\approx 0\\
  \boldsymbol{\mathcal{C}}_{3}&:=\boldsymbol{\pi}\,\approx \bm{0} \\
  \mathcal{C}_{4}&:=E-E_0\approx0 \, .
\end{align}
The equations of motion are given by
\begin{align}
  \dot{\tilde{\varphi}}&=\eta_{2} & \dot{\tilde{Q}}&=0\\
  \dot{\boldsymbol{\beta}}&=2\eta_{1}\boldsymbol{D}-\eta_{2}\boldsymbol{x} & \dot{\boldsymbol{D}}&=0\\
  \dot{\boldsymbol{x}}&=\boldsymbol{\eta}_{3} & \dot{\boldsymbol{\pi}}&=\eta_{2}\boldsymbol{D}\\
  \dot{t}&=-N & \dot{E}&=0 \, .
\end{align}
The time derivative of the constraints do not generate secondary
constraints. However, it imposes conditions on certain Lagrange
multipliers, indicating that some of the constraints are of
second-class. Indeed, the time derivative of $\bm{\mathcal{C}}_{3}$
yields $\eta_{2}=0$, while the time derivative of $\mathcal{C}_2$
gives
\begin{align}
  \boldsymbol{D}\cdot\dot{\boldsymbol{x}}=\boldsymbol{D}\cdot\boldsymbol{\eta}_{3}=0 \, .
\end{align}
Thus, the velocity component in the direction of the dipole moment
vanishes, a key characteristic of a planon. Moreover, as we will show
below, the components of $\boldsymbol{x}$ perpendicular to the dipole
moment transform under gauge transformations, indicating
that the dipole can move freely within the plane orthogonal to its
dipole moment.

To complete the analysis, we will classify the constraints into first
and second-class constraints. The second-class constraints are
represented by the pair:
\begin{align}
  \chi_{1}&:= \mathcal C_1 = \tilde{Q}-\boldsymbol{D}\cdot\boldsymbol{x} & \chi_{2}&:=\boldsymbol{D}\cdot\boldsymbol{\pi} \, .
\end{align}
Indeed, they obey the following Poisson bracket:
\begin{align}
  \left\{ \chi_{1},\chi_{2}\right\} =-d_{0}^{2}\neq0.
\end{align}
On the other hand, the first-class constraints are given by
\begin{align}
  \mathcal C_1 = \left\Vert \boldsymbol{D}\right\Vert ^{2}-d_{0}^{2}\approx0\qquad\qquad\qquad \boldsymbol{\pi}^{T}:=\boldsymbol{\pi}-\frac{1}{d_{0}^{2}}\left(\boldsymbol{D}\cdot\boldsymbol{\pi}\right)\boldsymbol{D}\approx\bm{0}\,.
\end{align}
Notice that the constraint
$\boldsymbol{\mathcal C}_3 = \boldsymbol{\pi} \approx 0$ was split
between its second-class part $\chi_2$ and its first-class part
$\boldsymbol{\pi}^T$. As mentioned above, the first-class constraints
generate, via $G= \bm{\lambda}(t)\cdot \bm{\pi}^{T}$, gauge
transformations of the coordinates
$\delta_{\lambda}\x=\{ \x,G\}=\bm{\lambda}^{T}$, which are
perpendicular to the dipole moment.

Given the 16 canonical variables, 4 first-class
constraints, and 2 second-class constraints, the number of independent
degrees of freedom is $16-2\times4-2=6$, which matches the dimension
of the orbit.

The Lagrangian simplifies when the second-class constraints are
enforced to vanish strongly, $\chi_1 = \chi_2 = 0$. By further
applying the gauge-fixing conditions $\tau = t$ and
$\boldsymbol{x}^T = \bm{0}$, and solving some of the first-class
constraints, we obtain:
\begin{align}
  L\left[\boldsymbol{x},\tilde{\varphi},\boldsymbol{D},\boldsymbol{\beta}\right]=d_{0}x^{L}\dot{\tilde{\varphi}}+\boldsymbol{D}\cdot\dot{\boldsymbol{\beta}}-\eta_{1}\left(\left\Vert \boldsymbol{D}\right\Vert ^{2}-d_{0}^{2}\right),
\end{align}
where $x^{L}:=\frac{\boldsymbol{D}}{d_0}\cdot\boldsymbol{x}$.

In the context of the Bargmann group interpretation, this orbit is
discussed in Section~8.3.4 of \cite{Figueroa-OFarrill:2024ocf} and
corresponds to a massless orbit describing a galilean particle that
does not evolve in time. The orbit is defined instantaneously at a
fixed value of the physical time $t$. In contrast, the planon case has
a non-trivial time evolution.

\subsubsection{Orbit \#8' (dipole with momentum)}
\label{sec:orbit-8-dipole}

Let us consider the following representative:
\begin{align}
  \alpha=\sM\left(0,0,\boldsymbol{d},\boldsymbol{p},\boldsymbol{0},E_0,0\right),
\end{align}
where $\left\Vert \boldsymbol{d}\right\Vert ^{2}=d_{0}^{2}$ and $\boldsymbol{d}\cdot\boldsymbol{p}=0$, with
$\left\Vert \boldsymbol{d}\times\boldsymbol{p}\right\Vert =d_{0}p_{0}>0$.

The Lagrangian is given by
\begin{align}
  L\left[R\left(\boldsymbol{\phi}\right),\boldsymbol{x},\boldsymbol{\beta},\tilde{\varphi},t\right]
  =\left(R\boldsymbol{d}\right)\cdot\boldsymbol{x}\dot{\tilde{\varphi}}+\left(R\boldsymbol{d}\right)\cdot\dot{\boldsymbol{\beta}}+R\boldsymbol{p}\cdot\dot{\boldsymbol{x}}-E_{0}\dot{t} \, .
\end{align}
It is convenient to introduce the following variables:
\begin{align}
  \boldsymbol{D}:=R\boldsymbol{d}\,\qquad\qquad\boldsymbol{\pi}:=R\boldsymbol{p}\,\qquad\qquad\tilde{Q}:=\boldsymbol{D}\cdot\boldsymbol{x},
\end{align}
The Lagrangian in Hamiltonian form can then be written as 
\begin{align}
  L\left[E,t,\tilde{Q},\tilde{\varphi},\boldsymbol{D},\boldsymbol{\beta},\boldsymbol{\pi},\boldsymbol{x},N,\eta_{1},\eta_{2},\eta_{3},\eta_{4}\right]
  =&-E\dot{t}+\tilde{Q}\dot{\tilde{\varphi}}+\boldsymbol{D}\cdot\dot{\boldsymbol{\beta}}+\boldsymbol{\pi}\cdot\dot{\boldsymbol{x}}-\eta_{1}\mathcal{C}_{1}-\eta_{2}\mathcal{C}_{2}  \nonumber \\
&-\eta_{3}\mathcal{C}_{3}-\eta_{4}\mathcal{C}_{4}
-N \mathcal C_5\,,
\end{align}
with the following constraints
\begin{align}
    \mathcal{C}_1&:=\tilde{Q}-\boldsymbol{D}\cdot\boldsymbol{x}\approx 0 \\
    \mathcal{C}_2&:=\left\Vert \boldsymbol{D}\right\Vert ^{2}-d_{0}^{2}\approx 0\\
    \mathcal{C}_{3}&:=\left\Vert \boldsymbol{\pi}\right\Vert ^{2}-p_{0}^{2}\approx 0 \\
    \mathcal{C}_{4}&:=\boldsymbol{D}\cdot\boldsymbol{\pi}\approx 0 \\
    \mathcal{C}_{5}&:=E-E_0\approx 0 \, .
\end{align}
Note that the constraints $\mathcal{C}_1$ and $\mathcal{C}_4$ are of second-class. Indeed, the obey the following Poisson bracket:
\begin{align}
  \left\{\mathcal{C}_{1},\mathcal{C}_{4}\right\} =-d_{0}^{2}\neq 0 \, .
\end{align}
The preservation in time of these second-class constraints implies that the corresponding Lagrange multipliers must satisfy
\begin{align}
  \eta_1=\eta_4=0 \, .
\end{align}
Consequently, the equations of motion become
\begin{align}
    \dot{t}&=-N &  \dot{E}&=0 \\
    \dot{\tilde{\varphi}}&=0 & \dot{\tilde{Q}}&=0  \\
    \dot{\boldsymbol{\beta}}&=2\eta_{2}\boldsymbol{D} &  \dot{\boldsymbol{D}}&=0 \\
    \dot{\boldsymbol{x}}&=2\eta_{3}\boldsymbol{\pi} & \dot{\boldsymbol{\pi}}&=0\,\label{eq:Pos}.
\end{align}

An immediate consequence to note is that
\begin{align}
  \boldsymbol{D}\cdot\dot{\boldsymbol{x}}=0 \, .
\end{align}
Therefore, the elementary dipole cannot move in the direction of its
dipole moment. It behaves as a planon.

The second-class constraints $\mathcal{C}_1=0$ and $\mathcal{C}_4=0$
can be solved by decomposing $\boldsymbol{\pi}$ into its longitudinal
and transverse components relative to the dipole moment. Consequently,
one can write
\begin{align}
  \tilde{Q}=d_{0}x^{L}\qquad\qquad\qquad \pi^{L}=0 \, ,
\end{align}
where $x^{L}:=\frac{\boldsymbol{D}}{d_0}\cdot\boldsymbol{x}$ and
$\pi^{L}:=\frac{\boldsymbol{D}}{d_0}\cdot\boldsymbol{\pi}$. If, in
addition, the gauge condition $\tau = t$ is imposed, the Lagrangian
can be expressed as
\begin{align*}
L\left[\tilde{\varphi},\boldsymbol{D},\boldsymbol{\beta},\boldsymbol{\pi},\boldsymbol{x},\eta_{2},\eta_{3}\right]=&d_{0}x^{L}\dot{\tilde{\varphi}}+\boldsymbol{D}\cdot\dot{\boldsymbol{\beta}}+\boldsymbol{\pi}^{T}\cdot\dot{\boldsymbol{x}}-\eta_{2}\left(\left\Vert \boldsymbol{D}\right\Vert ^{2}-d_{0}^{2}\right)-\eta_{3}\left(\left\Vert \boldsymbol{\pi}\right\Vert ^{2}-p_{0}^{2}\right) \,  .
\end{align*}
This orbit was analysed in the context of the Bargmann group in
Section 8.3.5 of \cite{Figueroa-OFarrill:2024ocf} and corresponds
to a massless representation without time evolution.

Let us study in detail the motion in the plane orthogonal to the dipole moment. Denoting its coordinates by $x$ and $y$, and assuming that $\pi_{x},\pi_{y}\neq 0$, from Eq.~\eqref{eq:Pos} it follows that
\begin{equation}
    \frac{\dot{x}}{\pi_{x}}=\frac{\dot{y}}{\pi_{y}}.
\end{equation}
Therefore, integrating over time, one obtains the following trajectory:
\begin{equation}
    y=\frac{\pi_y}{\pi_x}x+\bar{y},
\end{equation}
where $\bar{y}$ is a constant. Note that the above equation is gauge invariant, as it does not depend on $\eta_3$, and it shows that the particle moves in a straight line within the plane orthogonal to the dipole moment, which is consistent with the result obtained in Section \ref{sec:particle-dynamics}. On the other hand, when either $\pi_x$ or $\pi_{y}$ vanishes, the trajectory is also a straight line, aligned with the direction of the non-vanishing momentum.

\section{Composite dipoles}
\label{sec:composite-dipoles}

In this section, we construct a Lagrangian that describes the dynamics
of a non-elementary dipoles\footnote{To avoid confusion let us note
  that the composite dipoles we discuss are not elementary systems of
  the planon group. They are ideal or point electric dipoles which can
  be thought of as limits of the physical dipoles where the distance
  goes to zero and the charge to infinity such that the dipole moment
  stays finite (see, e.g.,~\cite{griffiths1999introduction}).} that
naturally couple to fracton traceless gauge
field~\cite{Pretko:2016lgv}, analyse their symmetries and first
quantise them.

First, we derive the Lagrangian for the dipole based solely on the
symmetries of the problem. The action has planon symmetries, but also
transversal galilean and longitudinal carrollian symmetries. We also
show how to derive the dynamics of these non-elementary dipoles by
assembling elementary monopoles (excitations of orbit \#6) by using an
appropriate interaction potential. These dipoles can be coupled to a
fracton traceless gauge field~\cite{Pretko:2016lgv} and we thereby
recover the generalised force law of Pretko from a variational
symmetry.

We also first quantise the resulting system resulting in a gaussian
Schrödinger-like field theory that was already discussed in the
context of the fracton/elasticity
duality~\cite{Kumar_2019,Pretko:2019omh}. We also show that the free
theory has again the mixed Carroll--Galilei symmetries.

\subsection{Dynamics of non-elementary dipoles from symmetries}
\label{sec:dynam-non-elem}

The charge density of a dipole with position
$\boldsymbol{z}(t)$ and dipole moment $\boldsymbol{d}$
takes the form
\begin{align}
  \label{eq:charge_density_dipole}
  \rho(t,\bm{x})&=-d^{i}\frac{\pd}{\pd x^{i}}\delta\left(\boldsymbol{x}-\boldsymbol{z}(t)\right) &
  J_{ij}(t,\bm{x})&=-d_{(i}\dot{z}_{j)}\delta\left(\boldsymbol{x}-\boldsymbol{z}(t)\right).
\end{align}
As discussed in Section~\ref{sec:introduction} the generators
$\boldsymbol{D}=\bm{d}$ and $Z=\bm{d}\cdot \zb$ depend exclusively on
the canonical variables $\boldsymbol{d}$ and $\boldsymbol{z}$. As a
result, the symmetries only act on their corresponding canonical
conjugates $\boldsymbol{\sigma}$ and $\boldsymbol{\pi}$. Using the
following canonical Poisson brackets
\begin{align}
  \label{eq:dip-poisson}
  \{ z^{i},\pi_{j}\} &=\delta^i_j & \{ d^{i},\sigma_{j}\} &=\delta^i_j \, ,
\end{align}
it can be shown that the transformation laws under quadrupole transformations
with parameter $\epsilon$ and dipole transformations with parameter
$\boldsymbol{\alpha}$ are given by
\begin{subequations}
  \label{eq:dipdpsym}
  \begin{align}
    \delta\boldsymbol{z} & =\bm{0} & \delta\boldsymbol{\pi}&=-\epsilon\boldsymbol{d}\label{eq:transf1}\\
    \delta\boldsymbol{d} & =\bm{0} & \delta\boldsymbol{\sigma}&=-\epsilon\boldsymbol{z}-\boldsymbol{\alpha} \,  .\label{eq:transf2}
  \end{align}
\end{subequations}
It is useful to decompose the momentum
$\boldsymbol{\pi}=\boldsymbol{\pi}^{T}+\boldsymbol{\pi}^{L}$ into its
transverse and longitudinal components relative to the dipole moment
$\boldsymbol{d}$
\begin{align}
  \pi_{i}^{T}&=P_{ij}\pi^{j} & \pi_{i}^{L}& = P_{ij} ^L \pi^j \,,
\end{align}
where we used the transverse and longitudinal projectors
\begin{align}
  \label{eq:proj}
  P_{ij}&=\delta_{ij}-\hat{d}_{i}\hat{d}_{j} & P_{ij}^{L} =\hat{d}_{i}\hat{d}_{j}   \, ,
\end{align}
with $\hat{d}_{i}$ being the unit vector in the direction of the
dipole moment $\hat{\boldsymbol{d}}=\boldsymbol{d}/d$, and $d$ is its
magnitude. They are orthogonal $P_{ij}P_{jk}^{L}= 0$.  In particular,
it is evident that only the longitudinal component of the momentum
transforms under \eqref{eq:transf1}
\begin{align}
  \delta\boldsymbol{\pi}^{T}& =\bm{0} & \delta\boldsymbol{\pi}^{L}&=-\epsilon\boldsymbol{d} \, .
\end{align}
Thus, to construct a Lagrangian for a free dipole that remains
invariant under planon symmetry, one may choose a Hamiltonian
proportional to $(\boldsymbol{\pi}^{T})^{2}=\pib^{2} - (\hat{\bm{d}}\cdot \pib)^{2}$. This choice yields the
following Lagrangian (in Hamiltonian form)
\begin{equation}
  \label{eq:Ldip}
  L_{\text{dip}}\left[\boldsymbol{z},\boldsymbol{\pi},\boldsymbol{\sigma},\boldsymbol{d}\right]=\pi_{i}\dot{z}^{i}+\sigma_{i}\dot{d}^{i}-\frac{1}{2m}P_{ij}\pi^{i}\pi^{j} \,  .
\end{equation}
The corresponding equations of motion are then given by
\begin{subequations}
  \begin{align}
    \dot{z}^{i} & =\frac{1}{m}P^{ij}\pi_{j}\label{eq:EOMb1}\\
    \dot{\pi}_{i} & =0\label{eq:EOMb2}\\
    \dot{d}^{i} & =0\label{eq:EOMb3}\\
    \dot{\sigma}_{i} & =\frac{\hat{\boldsymbol{d}}\cdot\boldsymbol{\pi}}{md}P_{ij}\pi^{j} \, .\label{eq:EOMb4}
  \end{align}  
\end{subequations}
Note that contracting equation \eqref{eq:EOMb1} with $d_{i}$ yields
\begin{align}
  \boldsymbol{d}\cdot\dot{\boldsymbol{z}}=0 \, ,
\end{align}
which is precisely the condition that ensures the conservation of the
trace of the quadrupole moment, cf.~\eqref{eq:mob_rest_dipoles}, and defines
the property of being a planon. Let us also emphasise that the
conservation of the dipole moment, $\dot{\bm{d}}=\bm{0}$, is a
consequence of the action principle and not put in by
hand.
The (im)mobility is particularly transparent when we also decompose the
position of the dipole into transversal and longitudinal components
$\zb = \zb^{T} + \zb^{L}$ and use equations \eqref{eq:EOMb1} and
\eqref{eq:EOMb2} to obtain
\begin{align}
  \label{eq:simpl}
  \ddot{\zb}^{T} &=\bm{0} & \dot{\zb}^{L} &=\bm{0} \, .
\end{align}
We see that the dipole can move freely transversal to the dipole
moment, but the longitudinal component is fixed (see
Figure~\ref{fig:planon}).

\subsection{Composite dipole from the interaction of two elementary monopoles}
\label{sec:comp-dipole-from}

In this section, we shall derive the Lagrangian for the composite
dipole \eqref{eq:Ldip} by considering the interaction of two
elementary monopoles described by
the orbit \#6. The interaction term is selected to ensure the
invariance of the Lagrangian under the diagonal subgroup of the direct
sum of the planon algebras associated with each elementary particle.

Let us consider the Lagrangian for two elementary particles of the
orbit \#6 with opposite electric and quadrupole charges. According
to \eqref{eq:monopole_Lag}, it can be written as
\begin{equation}
  L_{\text{dip}}=\boldsymbol{\pi}_{1}\cdot\dot{\boldsymbol{x}}_{1}+\boldsymbol{\pi}_{2}\cdot\dot{\boldsymbol{x}}_{2}-V\, .
\end{equation}
Here, $V$ represents the potential that enables the dipole to exist
as a bound state of two elementary particles. 

According to \eqref{eq:mom_mon}, the canonical momenta are defined as
\begin{align}
  \boldsymbol{\pi}_{1}&=q_{0}\left(\boldsymbol{x}_{1}\tilde{\phi}_{1}+\boldsymbol{b}_{1}\right) & \boldsymbol{\pi}_{2}&=-q_{0}\left(\boldsymbol{x}_{2}\tilde{\phi}_{2}+\boldsymbol{b}_{2}\right)\, .
\end{align}
In particular, under dipole transformations with parameters
$\boldsymbol{\alpha}_{1},\boldsymbol{\alpha}_{2}$ and quadrupole
transformations with parameters
$\epsilon_{1},\epsilon_{2}$, the canonical variables
transform according to
\begin{subequations}
  \begin{align}
    \delta\boldsymbol{x}_{1} & =0\qquad\qquad\delta\boldsymbol{\pi}_{1}=q_{0}\left(\boldsymbol{\alpha}_{1}+\epsilon_{1}\boldsymbol{x}_{1}\right)\\
    \delta\boldsymbol{x}_{2} & =0\qquad\qquad\delta\boldsymbol{\pi}_{2}=-q_{0}\left(\boldsymbol{\alpha}_{2}+\epsilon_{2}\boldsymbol{x}_{2}\right)\, .
  \end{align}
\end{subequations}
It is convenient to re-express the fields in the Lagrangian in terms
of the following quantities:
\begin{subequations}
  \begin{align}
    \boldsymbol{z} & :=\frac{1}{2}\left(\boldsymbol{x}_{1}+\boldsymbol{x}_{2}\right) & \boldsymbol{\pi}&:=\boldsymbol{\pi}_{1}+\boldsymbol{\pi}_{2}\\
    \boldsymbol{d} & :=q_{0}\left(\boldsymbol{x}_{1}-\boldsymbol{x}_{2}\right) & \boldsymbol{\sigma}&:=\frac{1}{2q_{0}}\left(\boldsymbol{\pi}_{1}-\boldsymbol{\pi}_{2}\right)\, .
  \end{align}
\end{subequations}
Here, $\boldsymbol{z}$ denotes the average position of the two
elementary particles, which describes the position of the dipole;
$\boldsymbol{\pi}$ represents the total momentum; and $\boldsymbol{d}$
is the dipole moment of the system. The variable
$\boldsymbol{\sigma}$, conjugate to $\boldsymbol{d}$, corresponds to
the momentum difference between the two elementary particles. A pure
dipole corresponds to the limit where the distance between the
elementary particles approaches zero
(\(\left\Vert\boldsymbol{x}_{1}-\boldsymbol{x}_{2}\right\Vert
\rightarrow 0\)) while their electric charge tends to infinity
(\(q_0 \rightarrow \infty\)), such that the product remains constant
in this limit.

The Lagrangian can then be rewritten as 
\begin{equation}
  L_{\text{dip}}=\boldsymbol{\pi}\cdot\dot{\boldsymbol{z}}+\boldsymbol{\sigma}\cdot\dot{\boldsymbol{d}}-V \,.
\end{equation}
Note that under the diagonal subgroup, with parameters $\epsilon=-\frac{1}{2}\left(\epsilon_{1}+\epsilon_{2}\right)$
for the quadrupole transformations, and $\boldsymbol{\alpha}=-\frac{1}{2}\left(\boldsymbol{\alpha}_{1}+\boldsymbol{\alpha}_{2}\right)$
for the dipole transformations, the fields transform as follow
\begin{subequations}
  \label{eq:planon-diag}
  \begin{align}
    \delta\boldsymbol{z} & =0 & \delta\boldsymbol{\pi}&=-\epsilon\boldsymbol{d}\\
    \delta\boldsymbol{d} & =0 & \delta\boldsymbol{\sigma}&=-\epsilon\boldsymbol{z}-\boldsymbol{\alpha} \, .
  \end{align}
\end{subequations}
These are precisely the transformations laws for the composite dipole
that were derived in \eqref{eq:dipdpsym}.

The potential must be selected to ensure the invariance under the
diagonal planon subgroup, in particular under the
transformations~\eqref{eq:planon-diag}.\footnote{This is analog to,
  e.g., the gravitational potential between two massive galilean
  particles, where the non-interacting particles are independently
  translation invariant $\delta x_{1,2}=a_{1,2}$, but the potential
  $V ( \| x_1-x_2 \| )$ is only invariant under translations by the
  same parameter $\delta x_{1,2}=a$.} Thus, it can be defined as
\begin{equation}
  V=\frac{1}{2m}P_{ij}\pi^{i}\pi^{j}.
\end{equation}
This term precisely corresponds to that appearing in the Lagrangian of
the composite dipole in \eqref{eq:Ldip}. In terms of the original
variables of the elementary excitations, it is expressed as follows
\begin{equation}
  \label{eq:Pot2}
  V=\frac{1}{2m}\left[\delta_{ij}-\frac{\left(x_{1}^{i}-x_{2}^{i}\right)\left(x_{1}^{j}-x_{2}^{j}\right)}{\left\Vert \boldsymbol{x}_{1}-\boldsymbol{x}_{2}\right\Vert ^{2}}\right]\left(\pi_{1}^{i}+\pi_{2}^{i}\right)\left(\pi_{1}^{j}+\pi_{2}^{j}\right).
\end{equation}

It is important to emphasise that the presence of the potential is
essential for enabling the dipole to move. In the absence of this
potential, the system would only describe two immobile decoupled monopoles.
The non-trivial dynamics, and in particular the planeon property of
the composite system, is a direct consequence of the presence of this
particular potential.

\subsection{Symmetries of the composite dipole}
\label{sec:symm-comp-dipole}

The dipole Lagrangian~\eqref{eq:Ldip} has the following symmetries and
corresponding Noether charges
\begin{subequations}
  \label{eq:global-sym}
  \begin{align}
    \delta X_{i} &= \omega_{ij}X_{j}  & \bm{J}&=  \zb\times \pib + \bm{d} \times  \bm{\sigma} \\
    \delta z_{i} &= c_{i} &  \bm{P} &=\pib \\
    \delta \sigma_{i} &=-\alpha_{i} & \bm{D} &=\bm{d} \\
    \delta \pi_{i} &= -\epsilon d_{i} \, , \delta\sigma_{i} = - \epsilon z_{i}  &  Z&=\bm{d} \cdot \bm{z} \\
    \delta X_{i}&= \lambda \dot X_{i} &   H&=\frac{1}{2m} (\bm{\pi}^{T})^{2} \, ,
  \end{align}
\end{subequations}
where $\omega_{ij} = -\omega_{ji}$ and $X_{i}$ denotes all canonical
variables. Using the Poisson brackets~\eqref{eq:dip-poisson} leads to
the expected symmetry algebra \eqref{eq:planonalg}
\begin{subequations}
  \label{eq:planon-poisson}
  \begin{align}
    \{ J_{i} ,J_{j}\} &= - \varepsilon_{ijk}J_{k} \\
    \{ J_{i},P_{j}\}&=\varepsilon_{ijk} P_{k} \\
    \{ J_{i},D_{j}\}&=\varepsilon_{ijk} D_{k} \\
    \{P_{i},Z\}&=-D_{i} \, .
  \end{align}
\end{subequations}
In particular
\begin{align}
      \{P_{i},D_{j}\}&=0  & \{H,Z\} &=0
\end{align}
where the first relation shows that the charge of the dipole is indeed
zero and the second that this system does not make use of the
possible central extension \eqref{eq:new-central-bracket}. Let us
emphasise that the rotations and the corresponding conserved angular
momentum are of the whole system.

Indeed due to the second term in the Hamiltonian
\begin{align}
 H=\frac{1}{2m} \left[
  \pib^{2} - (\hat{\bm{d}}\cdot \pib)^{2}
  \right]
\end{align}
there is a distinguished rotation symmetry
$(\zb, \pib,\bm{d}) \mapsto (R\zb, R \pib,R\bm{d})$ which preserves
the dipole moment $R \hat{\bm{d}}=\hat{\bm{d}}$.
This symmetry leads to the conserved longitudinal component of the
angular momentum $L_{i}=P_{ij}^{L}J_{j}$. Additionally the dipoles
also have a galilean boost symmetry in the transverse direction with a
conserved center of mass, where both are explicitly given by
\begin{align}
  \label{eq:boost}
  \bm{L}&=\bm{z}^{T} \times \pib^{T} &
                                       \bm{K} &= m \zb^{T} -  t \pib^{T} \,.
\end{align}
Using these symmetries we find that the transverse part leads to a
Bargmann-like central extension of the transversal Galilei algebra
\begin{subequations}
  \label{eq:gal-effective}
  \begin{align}
    \{L_{i},P_{j} \} &= \varepsilon_{ikl} P_{kj}P_{l}^{T}\\
    \{L_{i},K_{j} \} &= \varepsilon_{ikl} P_{kj}K_{l}\\
    \{K_{i},P_{j} \} &= m P_{ij} \\
    \{K_{i},H \} &= P^{T}_{i} \, .
  \end{align}
\end{subequations}
where the projector $P_{ij}$ takes the role of the usual Kronecker
delta. This agrees with the intuition that the dipole can move like a
massive galilean particle in the plane transverse to the dipole
moment. Due to the presence of the projectors this is a Poisson, but
no Lie algebra.

In the longitudinal direction there is no rotation symmetry, but
\begin{align}
  \bm{C} &= m \zb^{L} & \bm{\sigma}^{L}
\end{align}
are conserved. In particular the first quantity can be seen as a
carrollian center of mass leading to a Carroll-like algebra
\begin{subequations}
  \begin{align}
    \{C_{i},P_{j} \} &= m P^{L}_{ij} \\
    \{C_{i},H \} &= 0  \\
    \{D_{i},\sigma^{L}_{j} \} &= P_{ij}^{L} \, .
  \end{align}
\end{subequations}
This agrees with the intuition that the dipole is unable to move in
the longitudinal direction.

Let us remark that there is also an action of the rotations $J_{i}$ on
the longitudinal and transverse charges, but as is geometrically clear
they do not leave their respective subspaces invariant, e.g.,
$\{J_{i},z_{j}^{T}\} = \varepsilon_{ikl}P_{jl}z_{k}$ and
$\{J_{i},z_{j}^{L}\} = \varepsilon_{ikl}P^{L}_{jl}z_{k}$.

\subsubsection{Reduced action}
\label{sec:reduced-action}

We can construct a reduced, but inequivalent, action principle by
solving the equation of motion for the dipole moment
$\dot \db=\bm{0}$, i.e., fixing the dipole moment to a constant
value. This leads to $\db$ being a background value which we do not
vary. The decomposition into longitudinal and transversal parts can
then be done without using the equations of motion, i.e.,
off-shell. We can then integrate out $\pib^{T}$ to obtain
\begin{equation}
  \label{eq:Ldipother}
  L_{\bm{d}}\left[\zb^{T},\zb^{L},\pib^{L}\right]=\frac{1}{2} m \Vert \dot \zb^{T}\Vert^{2} + \pib^{L} \cdot \dot{\zb}^{L} \, ,
\end{equation}
which leads directly to~\eqref{eq:simpl}. This Lagrangian is the sum
of a transverse galilean and a longitudinal carrollian particle
action, thus manifesting the mixed galilean and carrollian symmetries
we discussed in Section~\ref{sec:symm-comp-dipole}.

However, the dipole conservation is not a consequence of the
variational principle and the symmetries~\eqref{eq:global-sym}, in
particular global rotation symmetries are absent.

\subsection{Coupling of the dipole to a traceless fracton gauge field}
\label{sec:coupl-dipole-trac}

Let us consider the coupling of the dipole~\eqref{eq:Ldip} to the
fracton traceless gauge field $\left(\phi,A_{ij}\right)$ introduced
in~\cite{Pretko:2016lgv}. The interaction is described by the
following Lagrangian
\begin{equation}
  L_{\text{int}}=\int dt d^{3}x\,\left(\phi\rho-A_{ij}J^{ij}\right) \, .
\end{equation}
Here $\phi$ represents the scalar potential, while $A_{ij}$ is the
symmetric and traceless tensor potential. The charge and current
densities of the dipole are given by:
\begin{align}
  \rho\left(t,\boldsymbol{x}\right)&=-d^{i}\partial_{i}\delta^{3}\left(\boldsymbol{x}-\boldsymbol{z}\left(t\right)\right) &  J^{ij}&=-d^{\left\langle i\right.}\dot{z}^{\left.j\right\rangle }\delta^{3}\left(\boldsymbol{x}-\boldsymbol{z}\left(t\right)\right) \, .
\end{align}
In these expressions, the brackets $\left\langle \cdots\right\rangle $
denote the symmetric and traceless part of a rank two tensor, i.e.,
$X_{\left\langle ij\right\rangle
}=\frac{1}{2}\left(X_{ij}+X_{ji}\right)-\frac{1}{3}\delta_{ij}X^k{}_k$.
We consider only the traceless part of $J^{ij}$ because this is the
only part that couples to the traceless $A_{ij}$.

The total Lagrangian of the system is then given by
\begin{equation}
L=\pi_{i}\dot{z}^{i}+\sigma_{i}\dot{d}^{i}-\frac{1}{2m}P_{ij}\pi^{i}\pi^{j}+d^{i}\partial_{i}\phi(t,\boldsymbol{z}(t))+A_{ij}(t,\boldsymbol{z}(t))d^{\left\langle i\right.}\dot{z}^{\left.j\right\rangle } \, ,\label{eq:Ltot}
\end{equation}
which leads to the following equations of motion:
\begin{subequations}
  \label{eq:eom-coupled}
  \begin{align}
    \dot{z}^{i} & =\frac{1}{m}P^{ij}\pi_{j}\label{eq:EOMGF1}\\
    \dot{\pi}_{i} & = d^{j}\left(E_{ij}+\frac{1}{3}\delta_{ij}\partial^{2}\phi\right)+\left(F_{ijk} + \delta_{k[i} \partial^l A_{j]l}\right)\dot{z}^{j}d^{k}\label{eq:EOMGF2}\\
    \dot{d}^{i} & =0 \label{eq:EOMGF3}\\
    \dot{\sigma}_{i} &= \frac{\hat{\bm{d}}\cdot \pib }{md} P_{ij}\pi^{j}+\partial_{i}\phi+A_{ij}\dot{z}^{j} \, .\label{eq:EOMGF4}
  \end{align}
\end{subequations}
Here,
\begin{subequations}
\begin{align}
    E_{ij} &:= \partial_{i}\partial_{j}\phi-\frac{1}{3}\delta_{ij}\partial^{2}\phi - \dot{A}_{ij} \\
    F_{ijk} &:= \partial_{i}A_{jk}-\partial_{j}A_{ik} - \frac12 \delta_{ik} \partial^l A_{jl} + \frac12 \delta_{jk} \partial^l A_{il}
\end{align}
\end{subequations}
denote the
electric and magnetic fields, respectively, which are gauge-invariant quantities.

Remarkably, by taking the time derivative of \eqref{eq:EOMGF1} and
using \eqref{eq:EOMGF2} and \eqref{eq:EOMGF3}, one obtains
\begin{equation}
  m\ddot{z}^{i}= P^{ij}d^{k}\left(E_{jk}+F_{jlk}\dot{z}^{l}\right) ,
\end{equation}
which corresponds precisely to the Lorentz force found by
Pretko~\cite{Pretko:2016lgv} (where however it did not derive from an
action principle). Like in the previous reference, one can also
introduce the dualised magnetic field
$B_{ij} := \frac12 \varepsilon_{i}{}^{pq} F_{pqj} \,\Leftrightarrow\,
F_{ijk} = \varepsilon_{ij}{}^p B_{pk}$ which, like the electric field
$E_{ij}$ is also gauge-invariant, symmetric and
traceless.\footnote{Our definition for the magnetic field
  $B^\text{here}_{ij} = \varepsilon_{i}{}^{pq} (\partial_p A_{qj} -
  \frac12 \delta_{jp} \partial^k A_{qk})$ differs slightly from
  Pretko's magnetic field
  $B^\text{there}_{ij} = \varepsilon_{i}{}^{pq} \partial_p A_{qj}$,
  which is not symmetric. This has no bearing on the expression of the
  Lorentz force however, since the extra terms vanish by virtue of the
  equations of motion.} Using this new quantity, the Lorentz force
reads
\begin{equation}
    m \ddot{z}^{i}= P^{ij}d^{k}\left(E_{jk} + \varepsilon_{jl}{}^p \dot{z}^{l} B_{pk}\right) = d \left( \boldsymbol E_{eff} + \dot{\boldsymbol z} \times \boldsymbol B_{eff} \right)^i .
\end{equation}
where $E_{eff}^i = P^{ij} \hat d^k E_{jk}$ and
$B_{eff}^i = P^{ij} \hat d^k B_{jk}$.  Due to the presence of the
projector $P^{ij}$, the velocity can be non-vanishing only in the
direction transverse to the dipole moment. Indeed, contracting
\eqref{eq:EOMGF1} with $d^{i}$ yields $\dot{\boldsymbol{z}}^{L}=0$.

Let us emphasise that the Lagrangian \eqref{eq:Ltot} and
the equations of motion \eqref{eq:eom-coupled} are invariant under the
following gauge transformations\footnote{While checking the gauge
  invariance, the following identity is useful
  $\frac{d}{dt} \Lambda(t,\boldsymbol z) = \partial_t
  \Lambda(t,\boldsymbol z) + \dot z^i \partial_i \Lambda(t,\boldsymbol
  z)$.}
\begin{align}
  \label{eq:gauge-tr}
  \delta_{\Lambda}\phi&=\partial_{t}\Lambda &
   \delta_{\Lambda}A_{ij}&=\partial_{i}\partial_{j}\Lambda-\frac{1}{3}\delta_{ij}\partial^{2}\Lambda                                                                   
& \delta_{\Lambda}\pi^{i}&=\frac{1}{3}d^{i}\partial^{2}\Lambda & \delta_{\Lambda}\sigma_{i}&=\partial_{i}\Lambda \, .
\end{align}
In particular, by selecting
$\Lambda=-\frac{1}{2}\epsilon z_{j}z^{j}-\alpha_{j}z^{j}$, one obtains
\begin{align}
  \delta_{\Lambda}\pi^{i}& =-\epsilon d^{i} & \delta_{\Lambda}\sigma_{i}& =-\epsilon z_{i}-\alpha_{i} \, ,
\end{align}
which are precisely the transformations of $\boldsymbol{\pi}$ and
$\boldsymbol{\sigma}$ under quadrupole and dipole transformations.
Hence, the global symmetries \eqref{eq:dipdpsym} can be regarded as
``large gauge transformations.''

\subsection{First quantisation of the composite dipole with fracton
  gauge field}
\label{sec:first-quantisation}

In this subsection, we shall perform the first quantisation of the
composite dipole interacting with an external fracton gauge field.  To
this end, we shall formulate the Lagrangian \eqref{eq:Ltot} in
canonical form, making the invariance under arbitrary time
reparametrisations explicit. Following Dirac's approach, we shall then
impose the condition that the wave function be annihilated by the
constraint generating time reparametrisations, with the momenta
substituted by their corresponding quantum operators.

Since the interaction term depends explicitly on the velocity of the
dipole, the field $\pi_{i}$ is no longer the canonical conjugate to
the position.  Instead, it is given as
\begin{equation}
  p_{i}:=\pi_{i}+d^{j}A_{ij} \, ,\label{eq:momentum_shift}
\end{equation}
which still satisfies the canonical Poisson bracket $\left\{ z^{i},p_{j}\right\} =\delta^i_j$.
This situation is similar to that of a particle coupled to an external
electromagnetic field. 

In terms of $p_{i}$, the Lagrangian \eqref{eq:Ltot}
reads
\begin{equation}
  \label{eq:Ltot-2}
L=p_{i}\dot{z}^{i}+\sigma_{i}\dot{d}^{i}-\frac{1}{2m}P^{ij}\left(p_{i}-d^{k}A_{ik}\right)\left(p_{j}-d^{l}A_{il}\right)+d^{i}\partial_{i}\phi\, .
\end{equation}

To express the action in a form that is manifestly invariant under
time reparametrisations, we must treat the time $t$ as a canonical
variable and introduce a conjugate momentum $p_{t}$, such that
$\left\{ t,p_{t}\right\} =-1$.  Then, we can write
\begin{equation}
  \label{eq:Ltot-1-1}
  L=-p_{t}\frac{d}{d\tau}t+p_{i}\frac{d}{d\tau}z^{i}+\sigma_{i}\frac{d}{d\tau}d^{i}-N\left(p_{t}+d^{i}\partial_{i}\phi-\frac{1}{2m}P^{ij}\left(p_{i}-d^{k}A_{ik}\right)\left(p_{j}-d^{l}A_{il}\right)\right) \, ,
\end{equation}
where $\tau$ denotes the proper time parameter, and $N$ is the Lagrange
multiplier that enforces the constraint
\begin{equation}
  \label{eq:constDip}
  p_{t}+d^{i}\partial_{i}\phi-\frac{1}{2m}P^{ij}\left(p_{i}-d^{k}A_{ik}\right)\left(p_{j}-d^{l}A_{il}\right)\approx 0 \, ,
\end{equation}
which generates arbitrary time reparametrisations.

The wave equation is obtained by requiring that the constraint
annihilates the wave function
$\psi\left(t,\boldsymbol{z},\boldsymbol{d}\right)$, where the momenta
are replaced by the quantum operators
\begin{align}
\label{eq:quant_cond}
  p_{t}&=i\frac{\partial}{\partial t} & p_{i}&=-i\frac{\partial}{\partial z^{i}} & \sigma_{i}&=-i\frac{\partial}{\partial d^{i}}\,.
\end{align}
Thus, one obtains the following Schrödinger-like wave equation
\begin{equation}
  \left[i\frac{\partial}{\partial t}+d^{i}\partial_{i}\phi-\frac{1}{2m}P^{ij}\left(i\frac{\partial}{\partial z^{i}}+d^{k}A_{ik}\right)\left(i\frac{\partial}{\partial z^{j}}+d^{l}A_{jl}\right)\right]\psi\left(t,\boldsymbol{z},\boldsymbol{d}\right)=0 \, . \label{eq:Gen_Schr}
\end{equation}
Note that, since the field $\boldsymbol{\sigma}$ does not appear
explicitly in the constraint (\ref{eq:constDip}), there are no
derivatives with respect to the dipole moment in the wave
equation. Consequently, the dipole moment $\boldsymbol{d}$ can be
treated as a fixed parameter within the wave equation.

The above equation of motion arises from variation of the following Lagrangian density
\begin{subequations}
    \label{eq:first-full-act}
  \begin{align}
    \mathcal{L}[\psi,\psi^{*}] &=i\psi^{*} \mathcal D_{t}\psi + \frac{1}{2m} \psi^* P^{ij} \mathcal{D}_{i} \mathcal{D}_{j}\psi \\
        &=i\psi^{*}\partial_{t}\psi+\frac{1}{2m}\psi^{*}P^{ij}\partial_{i}\partial_{j}\psi+d^{i}\psi^{*}\psi\partial_{i}\phi+\frac{id^{j}}{2m}P^{ik}\left(\psi\partial_{k}\psi^{*}-\psi^{*}\partial_{k}\psi\right)A_{ij} \nonumber \\
        & \quad -\frac{1}{2m}d^{k}d^{l}P^{ij}A_{ik}A_{jl}\psi^{*}\psi \, ,
  \end{align}
\end{subequations}
where we introduced the covariant derivatives 
\begin{align}
  \label{eq:cov-der}
\mathcal{D}_t&:= \partial_t - i d^i \partial_i\phi &  \mathcal{D}_{i}&:= \partial_{i} - i A_{ik}d^{k} \, ,
\end{align}
with $\partial_{t}:=\partial/\partial t$ and
$\partial_{i}:=\partial/\partial z^{i}$.

This Lagrangian is invariant under gauge transformations of the form
\begin{align}
\delta_{\Lambda}\phi&=\partial_{t}\Lambda & \delta_{\Lambda}A_{ij}&=\partial_{i}\partial_{j}\Lambda-\frac{1}{3}\delta_{ij}\partial^{2}\Lambda & \delta_{\Lambda}\psi&=id^{i}\left(\partial_{i}\Lambda\right)\psi \,.
\end{align}
In particular, the covariant derivatives transform as
\begin{align}
  \delta_\Lambda (\mathcal D_t \psi) &= i d^j \left(\partial_j \Lambda \right) \mathcal D_t \psi &
  \delta_{\Lambda}(\mathcal{D}_{i}\psi) &=
\left[ id^{j}\left(\partial_{j}\Lambda\right)\mathcal{D}_{i} - \frac{1}{3}
d_{i}\pd^{2}\Lambda \right]\psi   \,.
\end{align}
Additionally, it is invariant under a
global $U\left(1\right)$ transformation $\delta\psi=i\alpha\psi$,
which is associated with the electric charge.

In the specific case of large gauge transformations of the form
$\Lambda=-\frac{1}{2}\epsilon\left\Vert \boldsymbol{z}\right\Vert
^{2}-\boldsymbol{\alpha}\cdot\boldsymbol{z}$, where $\epsilon$ is the
parameter of the trace charge and $\boldsymbol{\alpha}$ represents the
parameter for dipole transformations, the wave function transforms as
follows:
\begin{equation}
  \label{eq:diptracetransf}
  \delta\psi=-i\boldsymbol{d}\cdot\left(\boldsymbol{\alpha}+\epsilon\boldsymbol{z}\right)\psi \, .
\end{equation}
The free Schrödinger-like Lagrangian
\begin{equation}
  \label{eq:Lfree}
  \mathcal{L}_{\text{free}}=i\psi^{*}\partial_{t}\psi+\frac{1}{2m}\psi^{*}P^{ij}\partial_{i}\partial_{j}\psi \,  ,
\end{equation}
is invariant under the global transformations in
\eqref{eq:diptracetransf}.

The dispersion relation associated with the free Lagrangian
\eqref{eq:Lfree} takes the form
\begin{equation}
    \omega=\frac{1}{2m}P^{ij}k_{i}k_{j}\,.
\end{equation}
Here, $\omega$ represents the frequency, and $k_i$ denotes the
momentum of the excitation. It is important to note that the
components of the momentum $\boldsymbol{k}$ in the direction of the
dipole moment $\boldsymbol{\hat{d}}$ do not contribute to the energy.

In $2+1$ dimensions, the Lagrangian density \eqref{eq:first-full-act}
was previously derived in \cite{Kumar_2019,Pretko:2019omh} (see
also~\cite[Section III.B.3]{Gromov:2022cxa}), in the context of the
study of the fracton/elasticity duality. Their derivation of the
Lagrangian followed a different approach from ours, where the
effective dipole dynamics was obtained directly from a field theory,
rather than through the quantisation of the classical particle-like
model for a dipole as discussed in
Section~\ref{sec:coupl-dipole-trac}.

\subsection{Symmetries of the free planon field theory}
\label{sec:summ-plan-field}

Let us determine the global conserved charges of the Lagrangian
\eqref{eq:Lfree} that describes a free quantum dipole. The Noether
current associated with the global transformations
\eqref{eq:diptracetransf} is given by
\begin{align}
  j^{0}&= \boldsymbol{d}\cdot \left(\boldsymbol{\alpha}+\epsilon\boldsymbol{z}\right)\psi^{*}\psi & j^{i}&=-\frac{i}{2m}\boldsymbol{d}\cdot\left(\boldsymbol{\alpha}+\epsilon\boldsymbol{z}\right)P^{ij}\left(\psi^{*}\partial_{j}\psi-\psi\partial_{j}\psi^{*}\right) \, , 
\end{align}
where the current satisfies the continuity equation
$\partial_{t}j^{0}+\partial_{i}j^{i}=0$.  As a consequence, the
corresponding conserved charges are given by
\begin{align}
  Q\left[\boldsymbol{\alpha},\epsilon\right]&=\int d^{3}\boldsymbol{z}  j^{0}=\boldsymbol{\alpha}\cdot\boldsymbol{D}+\epsilon\,Z\,,
\end{align}
where the dipole moment $\boldsymbol{D}$ and the trace charge $Z$ are
given by
\begin{align}
  \label{eq:DZ}
\boldsymbol{D}&=\boldsymbol{d}\int d^{3}\boldsymbol{z}\,\psi^{*}\psi=\boldsymbol{d} &  Z&=\int d^{3}\boldsymbol{z}\,\left(\boldsymbol{d}\cdot\boldsymbol{z}\right)\psi^{*}\psi\,.
\end{align}
Note that we have applied the standard normalisation condition for the
wave function $\int d^{3}\boldsymbol{z}\,\psi^{*}\psi=1$.

A natural question is whether these conserved quantities can be
derived from a continuity equation of the form typically found in
fracton theories. Specifically, for a planon symmetry, one expects an
equation of the form
\begin{equation}
    \partial_{t}\rho+\partial_{i}\partial_{j}J^{ij}=0,
\end{equation}
where $\rho$ represents the electric charge density, and $J^{ij}$ is a
symmetric, traceless tensor describing the dipole current.
The answer is positive for a charge and current density of the form
\begin{equation}
    \rho=-d^{i}\partial_{i}\left(\psi^{*}\psi\right)\qquad\qquad\qquad J^{ij}=\frac{i}{2m}d^{(i}P^{j)k}\left(\psi^{*}\partial_{k}\psi-\psi\partial_{k}\psi^{*}\right).
\end{equation}
Note that this expression for the charge density resembles that of a
classical dipole, except that the $\delta$-function specifying the
position of the classical dipole is replaced by the probability
density of the quantum dipole.
In particular, the total electric charge vanishes
\begin{equation}
    Q=\int d^{3}\boldsymbol{z}\,\rho=-\int d^{3}\boldsymbol{z} \,\partial_{j}\left(d^{j}\psi^{*}\psi\right)=0\,,
\end{equation}
which is consistent with the fact that a dipole is electrically neutral. 

Additionally, this quantum system retains the mixed
Carroll--Galilei boosts present in its classical counterpart
discussed in Section \ref{sec:symm-comp-dipole}. However, a key
difference arises: the commutator between a longitudinal Carroll boost
and a transverse galilean boost no longer vanishes.

Under infinitesimal transverse Galilei boosts of the form
$\delta x^{i}=-P^{ij}v_{j}t$ the wave function transforms according to
\begin{align}
  \delta\psi=P^{ij}v_{i}\left(t\partial_{j}\psi-imz_{j}\psi\right),
\end{align}
where the second term corresponds to a change in the phase.  Applying
Noether's theorem, the transformation above yields the following
conserved quantity associated with galilean boosts
\begin{equation}
    K_{i}=\int d^{3}\boldsymbol{z}\,iP^{ij}\left(t\psi^{*}\partial_{j}\psi-imz_{j}\psi^{*}\psi\right).
\end{equation}

Under infinitesimal longitudinal Carroll boosts of the form
$\delta t=-P_{L}^{ij}v_{i}z_{j},$
the wave function transforms according to 
\begin{align}
  \delta\psi  =P_{L}^{ij}v_{i}z_{j}\dot{\psi} \, .
\end{align}
Therefore, the Noether charge corresponding to longitudinal Carroll boosts is given by
\begin{equation}
    C_{i}=\int d^{3}\boldsymbol{z}\,\frac{1}{2m}P_{L}^{ij}P^{kl}z_{j}\left(\partial_{k}\psi^{*}\right)\left(\partial_{l}\psi\right).
\end{equation}
Using the Poisson bracket $\left\{ \psi\left(z\right),\psi^{*}\left(z'\right)\right\} =\frac{1}{i}\delta\left(z-z'\right)$, one finds
\begin{align}
  \left\{ K_{i},C_{j}\right\} =X_{ij},
\end{align}
where 
\begin{equation}
  X_{ij}:=-iP^{il}P_{L}^{jm}\int d^{3}\boldsymbol{z}\,z_{m}\psi^{*}\partial{}_{l}\psi,
\end{equation}
defines a new conserved quantity.

On the other hand, the generators associated with spacetime translations take the expected form 
\begin{align}
H&=-\int d^{3}\boldsymbol{z}\,\frac{1}{2m}P^{ij}\psi^{*}\partial_{i}\partial_{j}\psi & P_{i}&=-\int d^{3}\boldsymbol{z}\,i\psi^{*}\partial_{i}\psi \, .
\end{align}

The case of spatial rotations differs from the particle case studied
in Section \ref{sec:symm-comp-dipole}, as the dipole moment $d_i$ is
now fixed. In other words, it is treated as a background field that is
not varied in the action, whereas in the particle case it corresponds
to a canonical variable.  As a result, its presence breaks the full
rotational symmetry down to the subgroup that leaves the dipole
direction invariant, i.e., those satisfying
$R \hat{\bm{d}} = \hat{\bm{d}}$.  Consequently, the generator of
rotations around $\hat{\bm{d}}$ is given by
\begin{align}
  L_{i}=-iP_{L}^{il}\varepsilon_{ljk}\int d^{3}\boldsymbol{z}\,z^{j}\psi^{*}\partial_{k}\psi\,.
\end{align}

To restore rotational symmetry around all axes, it is necessary to
allow for the rotation of the entire system, including the dipole
moment.  One possibility is to consider $\bm{d}$ as an additional
variable over which we integrate the Lagrangian
\eqref{eq:first-full-act}, i.e., the action is
$\int\mathcal{L}[\psi,\psi^{*}] dt d^{3}\bm{z} d^{3}\bm{d} $. Indeed,
according to the quantisation conditions \eqref{eq:quant_cond} and
\eqref{eq:Gen_Schr} the wavefunction should depend on all of these
coordinates. In this extended configuration space the dipole moment
will transform under rotations and rotational symmetry is
restored. This is reminiscent of systems with spin, where we sum over
all spins, but for the dipoles it is an integral rather than a sum. It
is an interesting perspective to consider on this extended space, but
henceforth we will again work with fixed $\bm{d}$.

Similar to the standard Schrödinger equation, this model also
possesses a dilatation symmetry
\begin{equation}
  \delta\psi=\lambda_{1}\left(2t\partial_{t}\psi+z_{k}\partial_{k}\psi+\frac{3}{2}\psi\right),
\end{equation}
whose charge is given by 
\begin{equation}
  \Delta_{1}=\int d^{3}\boldsymbol{z}\,\left(-\frac{t}{m}P^{ij}\psi^{*}\partial_{i}\partial_{j}\psi+iz_{k}\psi^{*}\partial_{k}\psi+\frac{3i}{2}\psi^{*}\psi\right).
\end{equation}
Interestingly, this system also possesses a second type of dilatation symmetry
\begin{equation}
  \delta\psi=\lambda_{2}\left(2t\partial_{t}\psi+P^{kl}z_{k}\partial_{l}\psi+\psi\right),
\end{equation}
with a conserved charge given by 
\begin{equation}
  \Delta_{2}=\int d^{3}\boldsymbol{z}\,\left[-\frac{t}{m}P^{ij}\psi^{*}\partial_{i}\partial_{j}\psi+iP^{kl}z_{k}\psi^{*}\partial_{l}\psi+i\psi^{*}\psi\right]\,.
\end{equation}

This does not constitute a complete analysis of all possible
symmetries of the Lagrangian \eqref{eq:Lfree}. However, the preceding
results already illustrate the richness of this class of models. Once
the system is coupled to a gauge field, most of these symmetries are
broken. In particular, transverse galilean and longitudinal Carroll
boosts no longer define symmetries of the fractonic gauge field
theory, due to its inherently aristotelian structure.

\section{Quantum planon particles}
\label{sec:quant-plan-part}

The next stage in the study of the planon group is the construction of
its unitary irreducible (projective) representations, which we may
identify with the elementary quantum particles of the planon group.
These are again honest unitary irreducible representations of the
centrally extended planon group.  As we saw in
Section~\ref{sec:planon-particles}, this group is a semidirect product
$G = B \ltimes A$, with $B$ the Bargmann group and $A$ the two-dimensional
abelian group with Lie algebra $\a = \left<H,W\right>$.  In principle,
we may apply the method of Mackey to classify its UIRs, as was used
for fractons in \cite{Figueroa-OFarrill:2023qty}, for instance,
\emph{provided} that the semidirect product is regular (see, e.g.,
\cite[Ch.~17.1]{MR0495836}).  This is proved in
Appendix~\ref{sec:regul-centr-extend}.

Since $A$ is abelian, its unitary irreducible representations are
one-dimensional and the unitary dual $\widehat A \cong \RR^2$,
associating with $(E,c) \in \RR^2$ the character $\chi : A \to \U(1)$
defined by $\chi(\exp(t H + w W)) = e^{i (Et + cw)}$.  The action of
$B$ on $A$ induces an action on $\widehat A$.   Only the generator $Z$
acts nontrivially and we see that
\begin{equation}
  \exp(\widetilde\varphi Z) \cdot (E,c) = (E + c \widetilde\varphi, c),
\end{equation}
following on from the coadjoint action $\ad^*_Z H = -W$.  The unitary
dual $\widehat A$ decomposes under the action of $B$ into the
following orbits:
\begin{itemize}
\item point-like orbits $\left\{(E,0) \right\}$ for every $E \in
  \RR$;
\item and one-dimensional orbits $\RR \times \{c\}$ for every $c
  \neq 0$.
\end{itemize}
The point-like orbits are stabilised by the full Bargmann group $B$,
whereas a point $(E,c)$, with $c\neq 0$, in a one-dimensional orbit
is stabilised by the subgroup whose Lie algebra is the span of
$L_{ab},P_a, D_a, Q$ and is isomorphic to the Carroll algebra.
The action of $G$ on the one-dimensional orbits is by translations,
and therefore the standard Lebesgue measure on the real line is
invariant.

We therefore simply induce UIRs of $G$ from characters $\chi \in
\widehat A$ and a UIR of its stabiliser subgroup $G_\chi$.  We must
distinguish between the two kinds of characters.

\subsection{UIRs for characters in point-like orbits }
\label{sec:char-point-like}

For $\chi$ a character in a point-like orbit, $G_\chi = B$ and hence the
corresponding UIRs of $G$ are in bijective correspondence with the
UIRs of the Bargmann group, which can be read off from
\cite[Table~6]{Figueroa-OFarrill:2024ocf} (for $n=3$, which we tacitly
assume), after translating the notation from Bargmann to planon
language.  There are five classes of such UIRs, which we now
summarise.  We use the notation $g := \sg(\varphi, \tilde\varphi,
\bbeta, \ba, R,s,w) = \sfb(\varphi, \tilde\varphi, \bbeta, \ba, R) \exp(s H + w W)$.

\subsubsection{$\emph{\Romannum{1}} (\ell,\tilde q, E)$}
\label{sec:I-ell-tilde-q-E}

These are finite-dimensional UIRs with Hilbert space $\eH = V_\ell$,
the complex spin-$\ell$ UIR of $\Spin(3)$.  We let $\rho_\ell$ denote
the representation map $\rho_\ell : \Spin(3) \to \GL(V_\ell)$.  The
representation $U : G \to U(\eH)$ is given by
\begin{equation}
  U(g) \psi = e^{i s E} e^{i \tilde q \tilde \varphi} \rho_\ell(R) \psi.
\end{equation}

\subsubsection{$\emph{\Romannum{2}} (\ell, q, \tilde q, E)$}
\label{sec:II-ell-q-tilde-q-E}

Here the Hilbert space is $\eH = L^2(\RR^3, V_\ell)$ relative to the inner
product
\begin{equation}
  (\psi_1, \psi_2) = \int_{\RR^3} d\mu(\bd) \left<\psi_1(\bd), \psi_2(\bd)\right>_{V_\ell},
\end{equation}
with $d\mu(\bd)$ the standard euclidean measure on $\RR^3$.  The 
representation $U : G \to U(\eH)$ is given by
\begin{equation}
  \left( U(g) \psi \right)(\bd) = e^{i s E} e^{i q\varphi} e^{i \tilde q \tilde \varphi} e^{i \bbeta \cdot \bd} \rho_\ell(R) \psi(R^{-1}(\bd + q \ba)).
\end{equation}

\subsubsection{$\emph{\Romannum{3}} (n, p, \tilde q, E)$}
\label{sec:III-n-p-tilde-q-E}

Here the Hilbert space is $\eH = L^2(S^2,\scrO(-n))$, which we think of as
complex-valued smooth (i.e., not necessarily holomorphic) functions
$\psi(z)$ of a stereographic coordinate $z$ for $S^2$, which are
square-integrable relative to the inner product
\begin{equation}
  (\psi_1, \psi_2) = \int_{\CC} \frac{2i dz \wedge d\bar  z}{(1+|z|^2)^2} \overline{\psi_1(z)} \psi_2(z).
\end{equation}
The representation $U : G \to U(\eH)$ is given by
\begin{equation}
  \left( U(g) \psi \right)(z) = e^{i s E} e^{i \tilde q \tilde \varphi} e^{i \ba \cdot \bpi(z)} \left( \frac{\eta + \bar\xi  z}{|\eta + \bar\xi z|} \right)^{-n} \psi\left(\frac{\bar\eta z - \xi}{\eta + \bar \xi z} \right),
\end{equation}
where
\begin{equation}
  \label{eq:R-in-Spin-3}
  R =
  \begin{pmatrix}
    \eta & \xi \\ -\bar\xi & \bar\eta
  \end{pmatrix}
\end{equation}
and
\begin{equation}
  \label{eq:pi-z}
  \bpi(z) = \frac{p}{1+|z|^2} \left( 2 \Re(z), 2 \Im(z), |z|^2-1\right).
\end{equation}

\subsubsection{$\emph{\Romannum{4}} (n,d,E)$}
\label{sec:IV-n-d-E}

Here the Hilbert space is $\eH = L^2(\RR \times S^2, \scrO(-n))$,
where the complex line bundle $\scrO(-n)$ on $S^2$ has been pulled
back to $\RR \times S^2$ via the cartesian projection $\RR \times S^2
\to S^2$.  As in the previous case we think of these as complex-valued
smooth functions $\psi(u,z)$, where $z$ is a stereographic coordinate for
$S^2$, which are square-integrable relative to the inner product
\begin{equation}
  (\psi_1, \psi_2) = \int_{\RR \times \CC} \frac{2i du \wedge dz \wedge d\bar z}{(1+|z|^2)^2} \overline{\psi_1(u,z)}\psi_2(u,z).
\end{equation}
The representation $U : G \to U(\eH)$ is given by
\begin{equation}
  \left( U(g) \psi\right)(u,z) = e^{i s E} e^{-i u \tilde \varphi} e^{-i \bbeta \cdot \bdelta(z)} \left( \frac{\eta + \bar\xi
      z}{|\eta + \bar\xi z|} \right)^{-n} \psi\left(u - \ba \cdot
    \bdelta (z), \frac{\bar\eta z - \xi}{\eta + \bar \xi z} \right),
\end{equation}
where $R$ is as in equation~\eqref{eq:R-in-Spin-3} and
\begin{equation}
  \label{eq:delta-z}
  \bdelta(z) = \frac{d}{1+|z|^2} \left(2 \Re(z), 2 \Im(z), |z|^2-1  \right).
\end{equation}

\subsubsection{$\emph{\Romannum{5}}_\pm (d,p^\perp,E)$}
\label{sec:V-pm-d-p-perp-E}

Here the Hilbert space is $\eH = \Pi_{\pm} L^2(\RR \times S^3,\CC)$,
where the projectors $\Pi_\pm$ are given by
\begin{equation}
  \left( \Pi_\pm \psi \right)(u,S) = \tfrac12 \left( \psi(u,S)  \pm \psi(u,-S)\right)
\end{equation}
where $S \mapsto -S$ sends a point on $S^3$ to its antipodal point or,
equivalently, if we identify $S^3$ with $\SU(2)$, then this is just
multiplying the matrix $S$ by $-1$.  These functions are
square-integrable relative to the inner product
\begin{equation}
  (\psi_1, \psi_2) = \int_{\RR \times S^3} du \wedge d\mu(S) \overline{\psi_1(u,S)} \psi_2(u,S),
\end{equation}
with $d\mu(S)$ a bi-invariant volume form on $\SU(2)$.
The representation $U : G \to U(\eH)$ is given by
\begin{equation}
  \left( U(g)\psi \right)(u,S) = e^{i s E} e^{i\ba \cdot S \p} e^{i \bbeta \cdot S \bd} e^{i(\tilde\varphi + u) \ba \cdot S \bd} \psi(u + \tilde \varphi, R^{-1}S),
\end{equation}
where the action of $S \in \SU(2)$ on vectors is via the covering homomorphism $\SU(2) \to \SO(3)$.

\subsection{UIRs for characters in one-dimensional orbits}
\label{sec:char-one-dimens}

For $\chi$ a character in a one-dimensional orbit, we obtain UIRs as
square-integrable functions from the real line (the one-dimensional
orbit of $\chi$) with values in a UIR of the Carroll group, which for
$n=3$ can be read off from \cite[Table~5]{Figueroa-OFarrill:2023qty}.
Such characters are determined by a nonzero $c \in \RR$ and we can
choose as an orbit representative the character $\chi$ corresponding
to $(0,c)$; that is,
\begin{equation}
  \chi(\exp(s H + w W)) = e^{i c w}.
\end{equation}
We must choose $\sigma : \RR \to G$ such that $\sigma(E) \cdot (0,c) =
(E,c)$ and we take $\sigma(E) = \exp(\frac{E}{c} Z)$, which is
well-defined since $c \neq 0$.  The generic group element can be
written as
\begin{equation}
  g := \sg(\varphi, \tilde\varphi, \bbeta, \ba, R,s,w) =
  e^{\tilde\varphi Z} \sfc(\varphi, \bbeta, \ba, R) \exp(s H + w W),
\end{equation}
where $\sfc(\varphi, \bbeta, \ba, R) = e^{\varphi Q} e^{\bbeta \cdot
  \bD} e^{\ba \cdot \bP} R$ is a general element of the Carroll
subgroup stabilising $\chi$.  Let $G_\chi = C \ltimes A$ denote the
stabiliser subgroup of $\chi$, where $C$ is the Carroll subgroup and
$A$ the abelian subgroup generated by $H,W$.  Let $\eH$ be a Carroll
UIR twisted by the UIR of $A$ defined by the character $\chi$ and let
$F : G \to \eH$ be a $G_\chi$-equivariant (Mackey) function; that is,
\begin{equation}
  F(g h ) = h^{-1} \cdot F(g) \qquad\forall g \in G, h \in G_\chi.
\end{equation}
We define the corresponding section $\psi$ of the (trivial) vector
bundle $\RR \times \eH$ over the real line by $\psi(E) =
F(\sigma(E))$, which transforms under $g \in G$ as
\begin{equation}
  (g \cdot \psi)(E) = F(g^{-1}\sigma(E)).
\end{equation}
A calculation shows that
\begin{multline}
  g^{-1} \sigma(E) = \sigma(E - c \tilde\varphi) \sfc(\ba \cdot \bbeta
  - \varphi - \tfrac12 (\tfrac{E}{c}-\tilde\varphi) \|\ba\|^2,
  -R^{-1}(\bbeta - (\tfrac{E}{c}-\tilde\varphi) \ba), -R^{-1} \ba,
  R^{-1})\\
  \times e^{- s H} e^{-(w + t(\frac{E}{c}-\tilde\varphi))W}.
\end{multline}
Inserting this into $F$ and using the $G_\chi$-equivariance, we arrive
at
\begin{equation}
  (g \cdot \psi)(E) = e^{i(cw + s (E - c \tilde\varphi))}
  U(\sfc(\varphi - \tfrac12 (\tfrac{E}{c}-\tilde\varphi)\|\ba\|^2,
  \bbeta - (\tfrac{E}{c}-\tilde\varphi)\ba, \ba, R)) \cdot \psi(E - c \tilde\varphi),
\end{equation}
which is unitary relative to the inner product
\begin{equation}
  (\psi_1, \psi_2) = \int_\RR dE \left<\psi_1(E),\psi_2(E)\right>_{\eH},
\end{equation}
where $\left<-,-\right>_{\eH}$ is the inner product on the UIR $\eH$.
It is now a simple matter to go through the Carroll UIRs in
\cite[§3.4]{Figueroa-OFarrill:2023qty} and translate the explicit
representations of the Carroll group to the planon language.  As was
the case already with the Carroll and Bargmann groups, some of the
induced representations of $G$ thus constructed appear at face value
to be carried by sections of infinite-rank Hilbert bundles over the
character orbits.  But in the same way that it was possible in the
Carroll and Bargmann cases to view such representations as carried by
sections of finite-rank vector bundles over an auxiliary space
fibering over the character orbit, we can do the same here.

\subsubsection{$\emph{\Romannum{6}}(\ell,c)$}
\label{sec:VI-ell-c}

These are induced from the finite-dimensional representations of the
Carroll group in which the rotations act on the complex spin-$\ell$
representation $V_\ell$ and the other generators act trivially.  Hence
$\eH = L^2(\RR,V_\ell)$ relative to the inner product
\begin{equation}
  (\psi_1,\psi_2) = \int_{\RR} dE \left<\psi_1(E),\psi_2(E)\right>_{V_\ell}
\end{equation}
and the generic element $g := \sg(\varphi, \tilde\varphi,
\bbeta, \ba, R,s,w) \in G$ acts on $\eH$ as
\begin{equation}
  (U(g) \psi)(E) =e^{i(cw + s (E - c\tilde\varphi))} \rho_\ell(R) \psi(E - c\tilde\varphi).
\end{equation}

\subsubsection{$\emph{\Romannum{7}}(\ell,q,c)$}
\label{sec:VII-ell-q-c}

The inducing representation is carried by $L^2(\RR^3,V_\ell)$ relative
to the inner product
\begin{equation}
  (\psi_1,\psi_2) = \int_{\RR^3} d\mu(\bd) \left<\psi_1(\bd),\psi_2(\bd)\right>_{V_\ell},
\end{equation}
with $d\mu(\bd)$ the standard euclidean volume form on $\RR^3$.  This
means that the induced representation can be reformulated as carried
by $\eH = L^2(\RR^4,V_\ell)$ relative to the inner product
\begin{equation}
  (\psi_1,\psi_2) = \int_{\RR^4} dE \wedge d\mu(\bd) \left<\psi_1(E,\bd),\psi_2(E,\bd)\right>_{V_\ell},
\end{equation}
with the generic group element $g \in G$ acting on $\eH$ via
\begin{multline}
  (U(g) \psi)(E,\bd) = e^{i(cw + s(E - c \tilde\varphi) + q(\varphi-\tfrac1{2c}(E-c\tilde\varphi)\|\ba\|^2) + (\bbeta - \tfrac1c(E-c\tilde\varphi)\ba)\cdot \bd)} \\
 \times \rho_\ell(R) \psi(E - c\tilde\varphi,R^{-1}(\bd + q \ba)).
\end{multline}

\subsubsection{$\emph{\Romannum{8}}(n,d,c)$}
\label{sec:VIII-n-d-c}

The inducing representation is carried by $L^2(S^2, \scrO(-n))$ which
we describe as in Section~\ref{sec:IV-n-d-E} in terms of smooth
functions of a stereographic coordinate $z$.  We may
reformulate this UIR as carried by $L^2(\RR \times S^2, \scrO(-n))$,
where $\scrO(-n)$ has been pulled-back to $\RR \times S^2$ via the
cartesian projection, relative to the inner product
\begin{equation}
  \label{eq:IP-on-L2-R-S2}
  \left(\psi_1,\psi_2\right) = \int_{\RR \times \CC} \frac{dE \wedge
    2i dz \wedge d \bar z}{(1+|z|^2)^2} \overline{\psi_1(E,z)} \psi_2(E,z)
\end{equation}
and the action of the generic element $g \in G$ is given by
\begin{multline}
  (U(g)\psi)(E,z) = e^{i(c w + s( E - c \tilde\varphi))} e^{i (\bbeta - \frac1c (E - c \tilde\varphi)\ba) \cdot \bdelta(z)}\\
  \times \left( \frac{\eta + \bar\xi z}{|\eta + \bar\xi z|} \right)^{-n} \psi\left( E-c \tilde\varphi, \frac{\bar\eta z - \xi}{\eta + \bar \xi z}\right),
\end{multline}
with $R$ as in equation~\eqref{eq:R-in-Spin-3} and $\bdelta(z)$ as in
equation~\eqref{eq:delta-z}.

\subsubsection{$\emph{\Romannum{9}}(n,p,c)$}
\label{sec:IX-n-p-c}

This is very similar to the above case, with the inducing UIRs related
by a Carroll automorphism.  We reformulate the UIR again as carried by
$L^2(\RR \times S^2, \scrO(-n))$, relative to the inner product in
equation~\eqref{eq:IP-on-L2-R-S2} and now the action of the generic
element $g \in G$ is given by
\begin{equation}
  (U(g)\psi)(E,z) = e^{i(c w + s ( E - c \tilde\varphi))} e^{i \ba \cdot \bpi(z)}\left( \frac{\eta + \bar\xi z}{|\eta + \bar\xi z|} \right)^{-n} \psi\left( E-c \tilde\varphi, \frac{\bar\eta z - \xi}{\eta + \bar \xi z}\right),
\end{equation}
with $R$ as in equation~\eqref{eq:R-in-Spin-3} and $\bpi(z)$ as in
equation~\eqref{eq:pi-z}.

\subsubsection{$\emph{\Romannum{10}}_\pm(n,d,p,c)$}
\label{sec:X-n-d-p-c}

This representation is carried by the same Hilbert space as in the
previous two cases: namely, $L^2(\RR \times S^2, \scrO(-n))$ relative
to the usual inner product given by
equation~\eqref{eq:IP-on-L2-R-S2}.  In these UIRs, the vectors $\p$
and $\bd$ are either parallel (corresponding to the $+$ sign) or
anti-parallel (corresponding to the $-$ sign).  The action of the
generic $g \in G$ is given by
\begin{multline}
  (U(g)\psi)(E,z) = e^{i(c w + s (E - c \tilde\varphi))} e^{i(\bbeta - \frac1c(E - c\tilde\varphi) \ba)\cdot \bdelta(z) \pm \ba \cdot \bpi(z))} \\
  \times \left( \frac{\eta + \bar\xi z}{|\eta + \bar\xi z|} \right)^{-n} \psi\left( E-c \tilde\varphi, \frac{\bar\eta z - \xi}{\eta + \bar \xi z}\right),
\end{multline}
where $R$ as in equation~\eqref{eq:R-in-Spin-3}, $\bpi(z)$ as in
equation~\eqref{eq:pi-z} and $\bdelta(z)$ as in equation~\eqref{eq:delta-z}.

\subsubsection{$\emph{\Romannum{11}}_\pm(d,p,\theta,c)$}
\label{sec:XI-n-d-p-theta-c}

Finally, we have the representation for which $\bd = (0,0,d)$
and $\p = p(\sin\theta,0,\cos\theta)$.  This is carried by
$L^2_\pm(\RR \times S^3,\CC)$, where $L^2_\pm(\RR \times S^3,\CC)
\subset L^2(\RR \times S^3,\CC)$ is the closed subspace consisting of
square-integrable functions $\psi$ such that $\psi(E,-S) = \pm
\psi(E,S)$, where $E \in \RR$ and $S \in S^3$, but thought of as a
matrix in $\SU(2)$.  The inner product is given by
\begin{equation}
  (\psi_1,\psi_2) = \int_{\RR \times S^3} dE \wedge d\mu(S)
  \overline{\psi_1(E,S)} \psi_2(E,S),
\end{equation}
where $d\mu(S)$ is a bi-invariant volume form on $\SU(2)$ or,
equivalently, the volume form corresponding to a round metric on
$S^3$.  The action of the generic element $g \in G$ is given by
\begin{equation}
  (U(g)\psi)(E,S) = e^{i(c w + s (E - c \tilde\varphi))} e^{i\left(\left(\bbeta -\frac1c (E - c \tilde\varphi) \ba\right)\cdot S \bd + \ba \cdot S \p\right)} \psi(E-c\tilde\varphi,R^{-1}S).
\end{equation}

\subsection{Summary}
\label{sec:summary}

It follows from the fact that the extended planon group $G = B \ltimes
A$ is a regular semi-direct product (see
Appendix~\ref{sec:regul-centr-extend}) that all UIRs are obtained via
the Mackey method and, since $A$ is abelian, Rawnsley's theorem 
\cite{MR387499} says that they can all be constructed by quantising
coadjoint orbits.  As usual there are a couple of caveats: not every
coadjoint orbit is quantisable (e.g., there are coadjoint orbits whose
angular momentum/helicity is not quantised) and a given coadjoint
orbit can have inequivalent quantisations (this happens, in
particular, when the stabiliser of the orbit is not connected).  Short
of actually quantising the coadjoint orbits, any correspondence
between UIRs and coadjoint orbits remains conjectural and
Table~\ref{tab:orbits-UIRs} illustrates our attempt at such a
correspondence.

\begin{table}[h!]
  \centering
    \caption{Coadjoint orbits and UIRs of the centrally extended planon group $G$}
    \label{tab:orbits-UIRs}
  \begin{adjustbox}{max width=\textwidth}
    \begin{tabular}{*{2}{>{$}l<{$}}}
      \multicolumn{1}{c}{UIR Label} & \multicolumn{1}{c}{Orbit} \\
      \toprule\rowcolor{blue!7}
      \Romanbar{I}(0,\tilde q, E) & 0 (\tilde q, E) \\
      \Romanbar{I}(\ell, \tilde q, E) & 2 (\ell,\tilde q, E) \\\rowcolor{blue!7}
      \Romanbar{II}(0,q,\tilde q, E) & 6 (q,\tilde q, E) \\
      \Romanbar{II}(\ell,q,\tilde q, E) & 8 (\ell, q, \tilde q, E) \\\rowcolor{blue!7}
      \Romanbar{III}(n,p,\tilde q, E) & 4 (h=n, p, \tilde q, E)  \\
      \Romanbar{IV}(n,d,E) & 6'(h=n,d,E)  \\\rowcolor{blue!7}
      \Romanbar{V}_\pm(d,p^\perp, E) & 8'(d,p^\perp,E) \\
      \midrule
      \Romanbar{VI}(0,c) & 2'(c) \\\rowcolor{blue!7}
      \Romanbar{VI}(\ell,c) & 4'(\ell, c) \\
      \Romanbar{VII}(0,q,c) & 8'''(q,c) \\\rowcolor{blue!7}
      \Romanbar{VII}(\ell,q,c) & 10(\ell, q, c) \\
      \Romanbar{VIII}(n,d,c) & 6_0'''(h=n,d,c) \\\rowcolor{blue!7}
      \Romanbar{IX}(n,p,c) & 6''(h=n,p,c) \\
      \Romanbar{X}_\pm (n,d,p,c) & 6'''_\pm(h=n, d, p, c) \\\rowcolor{blue!7}
      \Romanbar{XI}_\pm(d,p,\theta,c)& 8''(d,p,\theta,c) \\
      \bottomrule
    \end{tabular}
  \end{adjustbox}
  \vspace{1em}
  \caption*{The table lists for every UIR of the extended planon group
    $G = B \ltimes A$ the coadjoint orbit from which it can be
    obtained by geometric quantisation.  The horizontal line separates
    those UIRs which are essentially UIRs of Bargmann (above the line)
    from the rest.  The orbits of type $8'''$ and $10$ have $q=c$, but
    that was a choice.  It is in fact the ratio $\lambda = q/c$ which
    was set to $1$ via a rescaling of generators.  To make more
    transparent contact with the UIRs we should keep the ratio
    arbitrary (but nonzero).  The parameters $\ell$ and $c$, when they
    appear are assumed to be nonzero and $2\ell$ is a non-negative
    integer.  The parameter $n$ is an integer.  The two UIRs
    $\Romanbar{V}_\pm(d,p^\perp, E)$ and
    $\Romanbar{XI}_\pm(d,p,\theta,c)$ have signs which are not
    apparent in the corresponding coadjoint orbits.  This is due to
    the fact that the coadjoint orbits have stabilisers with two
    connected components and hence they admit two inequivalent
    quantisations.}
\end{table}

\section{Discussion}
\label{sec:discussion}

In this work, we have analysed both elementary and composite systems
exhibiting planon symmetry. We provided a classification of classical
and quantum elementary particles by characterising the coadjoint
orbits and unitary irreducible representations of the (extended) planon
group. Planon symmetries are closely related to Bargmann
symmetries~\cite{Gromov:2018nbv}, but their elementary systems are
mapped into each other in quite a nontrivial way. For example the
monopole, that is stuck to a point, corresponds to the massive
galilean particle, which moves along straight lines, while the
unrestricted dipoles correspond to massless galilean particles. In
this sense, we provide a physically interesting setup for these often
ignored galilean orbits.

We also present an action for composite dipole particles that
naturally couples to Pretko’s traceless scalar charge theory,
see~\eqref{eq:Ltot}. Remarkably, both the free particle and its
first-quantised theory exhibit a mixed Carroll–Galilei symmetry, which
is fully consistent with their allowed (im)mobility (see
Figure~\ref{fig:planon}).

Finally, we classify the planon elementary quantum systems, i.e., the
unitary irreducible representations of the extended planon group and
set up a correspondence between the UIRs and the coadjoint orbits.

There are various interesting avenues for further exploration:
\begin{description}
\item[Other multipole symmetries and particles] The tools we have
  developed here and
  in~\cite{Figueroa-OFarrill:2023vbj,Figueroa-OFarrill:2023qty,Perez:2023uwt}
  can be generalised to other multipole
  symmetries~\cite{Gromov:2018nbv} and their associated particles. A possible
  next step is to analyse lineons and vector charge
  theories~\cite{Pretko:2016lgv} (see also
  \cite{Bertolini:2022ijb,Baig:2023yaz,Kasikci:2023tvs,Ahmadi-Jahmani:2025iqc,Bergshoeff:2025qtt}).

\item[Central extension] In Section~\ref{sec:centr-extens-plan} we
  showed that the planon group admits a nontrivial central
  extension. In the context of this work, this leads to energy that is
  unbounded from below. However, it is possible that alternative
  interpretations could give rise to interesting new
  physics.\footnote{There are even simpler systems, such as
    aristotelian spacetimes with one scalar
    charge~\cite{Figueroa-OFarrill:2022kcd}, for which it would be
    interesting to find a physical realization.}
  
\item[Elementary versus composite] We discussed two types of
  particles: elementary particles of the planon group and composite
  dipoles constructed from elementary monopoles. While the composite
  dipoles naturally couple to the traceless scalar charge gauge
  theory, it remains an open question whether the elementary dipoles
  in~\eqref{eq:action-functional-general} also couple to a gauge
  theory.

  A related question is whether the composite dipoles can be regarded
  as elementary systems of other symmetries, maybe such as those discussed
  in Section~\ref{sec:composite-dipoles}.
  
\item[Soft theorems and the infrared triangle] The scalar tensor gauge
  theory~\cite{Perez:2022kax,Perez:2023uwt} exhibits an intricate
  interplay between memory effects, soft theorems, and asymptotic
  symmetries, collectively forming an infrared
  triangle~\cite{Strominger:2017zoo}.

  Following~\cite{Perez:2023uwt}, one could leverage the coupling
  between particles and the gauge field, \eqref{eq:Ltot}, along with
  the field theory framework~\eqref{eq:first-full-act}, to compute
  memory effects and soft theorems. Notably, the study of these
  phenomena in $2+1$ dimensions appears particularly promising due to
  potential applications in elasticity and vortex dynamics (see also
  \cite{Tsaloukidis:2023jmr}).

\item[Carroll in crystals] Motivated by the relation between fractons
  and elasticity, let us provide a simple example of how Carroll and
  Carroll--Galilei symmetries emerge in the context of defects in
  crystals.

  Consider a point defect in a lattice, which we can envision by just
  removing one atom. At zero temperature this defect is stuck to a
  point, which is described by the action~\eqref{eq:monopole_act} and
  which has Carroll symmetry. At higher temperatures one would need
  to introduce a potential that governs the now statistically allowed
  movement.

  Let us now restrict to two spatial dimensions and consider the
  movement of a vacancy under a strain, as depicted in
  Figure~\ref{fig:glide}.

  \begin{figure}
    \centering
    \tikzset{
  ball/.pic={\filldraw[line width=0.5pt, fill=white] circle [radius=2pt];}
}

\hspace{-160pt}
\begin{tikzpicture}
\tikz{
 \draw[line width = 1pt] (0,0) .. controls(0,1) and (0.2,2) .. (0.2,3)
 	foreach \t in {0.11, 0.22, ..., 0.9} {pic [pos=\t] {ball}};
 \draw[line width = 1pt] (0.4,0) .. controls(0.4,1) and (0.6,2) .. (0.6,3)
	foreach \t in {0.11, 0.22, ..., 0.9} {pic [pos=\t] {ball}};
 \draw[line width = 1pt] (0.8,0) .. controls(0.8,1) and (1,2) .. (1,3)
	foreach \t in {0.11, 0.22, ..., 0.9} {pic [pos=\t] {ball}};
 \draw[line width = 1pt] (-0.2,1.5) .. controls(-0.2,2) and (-0.2,2.5) .. (-0.2,3)
	foreach \t in {0.1, 0.32, ..., 0.9} {pic [pos=\t] {ball}};
 \draw[line width = 1pt] (-0.4,0) .. controls(-0.4,1) and (-0.6,2) .. (-0.6,3)
	foreach \t in {0.11, 0.22, ..., 0.9} {pic [pos=\t] {ball}};
 \draw[line width = 1pt] (-0.8,0) .. controls(-0.8,0.5) and (-0.85,1) .. (-0.9,1.5)
	foreach \t in {0.22, 0.44, ..., 0.9} {pic [pos=\t] {ball}};
 \pic[draw=red] at (-0.2,1.33) {ball};
 \draw[line width = 1pt] (4,0) .. controls(4,1) and (4.2,2) .. (4.2,3)
 	foreach \t in {0.11, 0.22, ..., 0.9} {pic [pos=\t] {ball}};
 \draw[line width = 1pt] (4.4,0) .. controls(4.4,1) and (4.6,2) .. (4.6,3)
	foreach \t in {0.11, 0.22, ..., 0.9} {pic [pos=\t] {ball}};
 \draw[line width = 1pt] (3.8,1.5) .. controls(3.8,2) and (3.8,2.5) .. (3.8,3)
	foreach \t in {0.1, 0.32, ..., 0.9} {pic [pos=\t] {ball}};
 \draw[line width = 1pt] (3.6,0) .. controls(3.6,1) and (3.4,2) .. (3.4,3)
	foreach \t in {0.11, 0.22, ..., 0.9} {pic [pos=\t] {ball}};
 \draw[line width = 1pt] (2.8,0) .. controls(2.8,0.5) and (2.75,1) .. (2.7,1.5)
	foreach \t in {0.22, 0.44, ..., 0.9} {pic [pos=\t] {ball}};
 \draw[line width = 1pt] (3.2,0) .. controls(3.2,1) and (3,2) .. (3,3)
	foreach \t in {0.11, 0.22, ..., 0.9} {pic [pos=\t] {ball}};
 \pic[draw=red] at (3.8,1.33) {ball};
 \draw[line width = 1pt] (8,0) .. controls(8,1) and (8.2,2) .. (8.2,3)
 	foreach \t in {0.11, 0.22, ..., 0.9} {pic [pos=\t] {ball}};
 \draw[line width = 1pt] (7.8,1.5) .. controls(7.8,2) and (7.8,2.5) .. (7.8,3)
	foreach \t in {0.1, 0.32, ..., 0.9} {pic [pos=\t] {ball}};
 \draw[line width = 1pt] (7.6,0) .. controls(7.6,1) and (7.4,2) .. (7.4,3)
	foreach \t in {0.11, 0.22, ..., 0.9} {pic [pos=\t] {ball}};
 \draw[line width = 1pt] (6.4,0) .. controls(6.4,0.5) and (6.35,1) .. (6.3,1.5)
	foreach \t in {0.22, 0.44, ..., 0.9} {pic [pos=\t] {ball}};
 \draw[line width = 1pt] (7.2,0) .. controls(7.2,1) and (7,2) .. (7,3)
	foreach \t in {0.11, 0.22, ..., 0.9} {pic [pos=\t] {ball}};
 \draw[line width = 1pt] (6.8,0) .. controls(6.8,1) and (6.6,2) .. (6.6,3)
	foreach \t in {0.11, 0.22, ..., 0.9} {pic [pos=\t] {ball}};
 \pic[draw=red] at (7.8,1.33) {ball};
 \draw[line width = 1pt] (11.4,0) .. controls(11.4,1) and (11.4,2) .. (11.4,3)
 	foreach \t in {0.11, 0.22, ..., 0.9} {pic [pos=\t] {ball}};
 \draw[line width = 1pt] (11.8,1.5) .. controls(11.8,2) and (11.8,2.5) .. (11.8,3)
	foreach \t in {0.1, 0.32, ..., 0.9} {pic [pos=\t] {ball}};
 \draw[line width = 1pt] (11,0) .. controls(11,1) and (11,2) .. (11,3)
	foreach \t in {0.11, 0.22, ..., 0.9} {pic [pos=\t] {ball}};
 \draw[line width = 1pt] (9.8,0) .. controls(9.8,0.5) and (9.8,1) .. (9.8,1.5)
	foreach \t in {0.22, 0.44, ..., 0.9} {pic [pos=\t] {ball}};
 \draw[line width = 1pt] (10.2,0) .. controls(10.2,1) and (10.2,2) .. (10.2,3)
	foreach \t in {0.11, 0.22, ..., 0.9} {pic [pos=\t] {ball}};
 \draw[line width = 1pt] (10.6,0) .. controls(10.6,1) and (10.6,2) .. (10.6,3)
	foreach \t in {0.11, 0.22, ..., 0.9} {pic [pos=\t] {ball}};
 \pic[draw=red] at (11.8,1.33) {ball};
 \draw[line width = 0.6pt, dashed] (-1.2,1.5) -- (12.2,1.5);
 \draw[line width = 0.8pt, -Stealth] (-0.2,3.5) -- (0.4,3.5);
 \draw[line width = 0.8pt, -Stealth] (0.4,-0.5) -- (-0.2,-0.5);
 \draw[line width = 0.8pt, -Stealth] (3.4,3.5) -- (4,3.5);
 \draw[line width = 0.8pt, -Stealth] (4,-0.5) -- (3.4,-0.5);
 \draw[line width = 0.8pt, -Stealth] (7,3.5) -- (7.6,3.5);
 \draw[line width = 0.8pt, -Stealth] (7.6,-0.5) -- (7,-0.5);
 \draw[line width = 0.8pt, -Stealth] (10.6,3.5) -- (11.2,3.5);
 \draw[line width = 0.8pt, -Stealth] (11.2,-0.5) -- (10.6,-0.5);
 \draw[line width = 1pt, -Stealth, draw=red] (-2.2,1.3) -- (-1.4,1.3);
 \node at (-1.8,0.8) {{\color{red} $\vec{\bm{b}}$}};
 \node at (-1.5,3.5) {shear stress};
 \draw[-] (-1.2,1.5) -- node[above=2mm] {slip plane} (-2,2);
}
\end{tikzpicture}
    \caption{The glide motion of a dislocation (red circle) is
      parallel to the Burgers vector $\vec{\bm{b}}$.}
    \label{fig:glide}
  \end{figure}
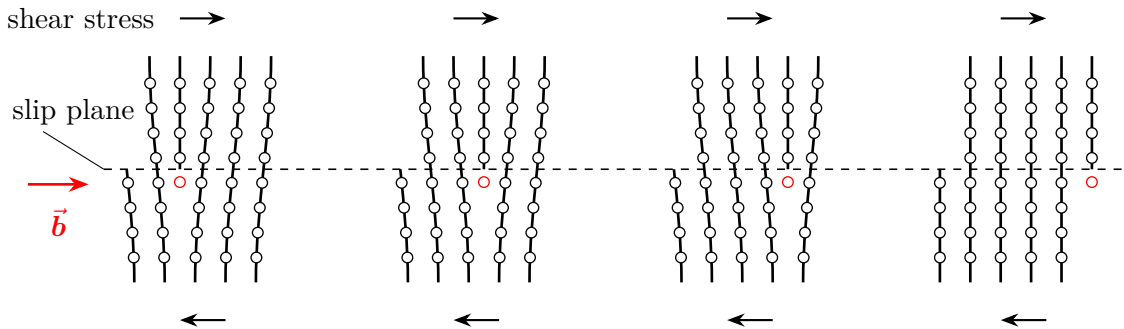
  The movement of the vacancy is restricted to be parallel to the
  Burgers vector. Once we identify the Burgers vector with the dipole
  vector via $b_{i}=\epsilon_{ij}d_{i}$ the movement of the vacancy
  follows~\eqref{eq:Spl}, i.e., is restricted to be orthogonal to the
  dipole moment. We have restricted to low temperature and (infinite)
  planar glide motion. For more information see,
  e.g.,~\cite{hull2011introduction}.
  
  In this sense defects in crystals provide the possibly simplest
  physically realised model with Carroll(--Galilei) symmetry. There
  could be a further interesting connections between nonlorentzian
  geometry and condensed matter physics that await exploration.
\end{description}

\section*{Acknowledgments}

The work of SiP was partially supported by the Fonds de la Recherche
Scientifique -- FNRS under Grant No. FC.36447, the SofinaBo\"el Fund
for Education and Talent and the \textit{Fonds Friedmann} run by the
\textit{Fondation de l'\'Ecole polytechnique}. Since 1 November 2024,
SiP is supported by European Research Council (ERC) Project 101076737
-- \textit{CeleBH}.

This work was initiated while StP was supported by the
Leverhulme Trust Research Project Grant (RPG-2019-218) “What is
Non-Relativistic Quantum Gravity and is it Holographic?”.

The research of AP is partially supported by Fondecyt grants No
1220910 and 1230853.

JMF would like to thank Jarah Fluxman for interesting conversations
about the Mackey method.

AP and StP acknowledge support from the Erwin Schrödinger
International Institute for Mathematics and Physics (ESI) where some
of the research was undertaken during the ESI Research in Teams
Programme.

\appendix

\section{Summary of \cite{Figueroa-OFarrill:2024ocf} in planon  language}
\label{sec:bargmann-to-planon-summary}

In this appendix we summarise the results of
\cite{Figueroa-OFarrill:2024ocf} on the Bargmann group in the planon
language.  This is mostly an exercise in translating notation, but we
think it would be useful to record the translation here.

The basic dictionary is the one relating the bases of the Lie
algebras, which as mentioned already translates the Bargmann basis
$\left<L_{ab},B_a, P_a, H, M\right>$ to
$\left<L_{ab},P_a, D_a, Z, Q\right>$ and is summarised in
Table~\ref{tab:planon-barg}.  The Bargmann canonical dual basis
$\left<\lambda^{ab}, \beta^a, \pi^a, \eta, \mu\right>$ translates to
$\left<\lambda^{ab}, \pi^a, \delta^a, \zeta,\theta\right>$ and hence
the Bargmann moment\footnote{In \cite{Figueroa-OFarrill:2024ocf} we
  chose the sign in front of $p_a \pi^a$ in order that under the boost
  with parameter $\bv$, the momentum $\p$ would change by
  $\p + m \bv$.  We have decided to keep the minus sign in the planon
  case for no other reason than to simplify the translation.}
\begin{equation}
  \label{eq:barg-moment}
  \tfrac12 J_{ab} \lambda^{ab}  + k_a \beta^a - p_a \pi^a + E \eta + m \mu
\end{equation}
translates to
\begin{equation}
  \tfrac12 J_{ab} \lambda^{ab}  + p_a \pi^a - d_a \delta^a + \tilde q
  \zeta + q \theta.
\end{equation}

\begin{table}[h!]
    \centering
    \begin{tabular}{l|l}
        Planon & Bargmann $\oplus$ $\mathbb R$ \\
        \hline
        $L_{ab}$: spatial rotations & $L_{ab}$: spatial rotations \\
        $P_a$: spatial translations & $B_a$: galilean boosts \\
        $D_a$: dipole moment & $P_a$: spatial translations \\
        $Z$: trace quadrupole moment & $H$: time translation \\
        $Q$: electric charge & $M$: mass \\
        $H$: time translation & $\mathbb R$
    \end{tabular}
    \caption{Correspondence between planon and Bargmann algebra generators.}
    \label{tab:planon-barg}
\end{table}

The Bargmann group parameters $(a_+,a_-,\ba, \bv, R)$ in
\cite{Figueroa-OFarrill:2024ocf} now become $(\varphi,\tilde\varphi,
\bbeta, \ba, R)$ and the Bargmann moments $(m,E, \p, \bk,
J)$ in \cite{Figueroa-OFarrill:2024ocf} now become $(q,\tilde q,
\bd,\ba, R)$.

The generic Bargmann group element in \cite{Figueroa-OFarrill:2024ocf}
is therefore now $\sfb(\varphi,\tilde\varphi,\bbeta,
\boldsymbol{a}, R)$.  Its coadjoint action on the moment $\sM(J, 
\tilde q, q, \p, \bd)$ is then given by
\begin{equation}
  \Ad^*_{\sfb(\varphi,\tilde\varphi,\bbeta,\boldsymbol{a},R)}
  \sM(q, \tilde q,\bd, \p, J) = \sM(q', \tilde q',\bd', \p', J'),
\end{equation}
where
\begin{equation}
  \label{eq:barg-coadjoint}
  \begin{split}
    J' &= RJR^T + \bbeta (R \bd + q \boldsymbol{a})^T - (R \bd + q  \boldsymbol{a}) \bbeta^T + (R \p) \ba^T - \ba (R  \p)^T\\
    \p' &= R\p + q \bbeta + \tilde\varphi (R  \bd + q \ba)\\
    \bd' &= R \bd + q \ba\\
    \tilde q' &= \tilde q + R \bd \cdot \ba + \tfrac12 q  \|\bd\|^2\\
    q' &= q.
  \end{split}
\end{equation}
In the particular case of the spatial dimension $n=3$, the moment
$\sM(q, \tilde q,\bd, \p, \bj)$ transforms into
\begin{equation}
  \label{eq:barg-coadjoint-3d}
  \begin{split}
    \bj' &= R\bj - \bbeta \times (R \bd + q \boldsymbol{a}) + \ba \times R \p\\
    \p' &= R\p + q \bbeta + \tilde\varphi (R  \bd + q \ba)\\
    \bd' &= R \bd + q \ba\\
    \tilde q' &= \tilde q + R \bd \cdot \ba + \tfrac12 q  \|\bd\|^2\\
    q' &= q.
  \end{split}
\end{equation}
Table~1 in \cite{Figueroa-OFarrill:2024ocf} translates into Table~\ref{tab:barg-coad-orbs-planon}.

\begin{table}[h]
  \centering
    \caption{Coadjoint orbits of the Bargmann group $B$ in planon language}
    \label{tab:barg-coad-orbs-planon}
    \setlength{\extrarowheight}{3pt}
  \resizebox{\linewidth}{!}{
    \begin{tabular}{*{3}{>{$}l<{$}}>{$}c<{$}>{$}l<{$}}
      \multicolumn{1}{l}{\#} & \multicolumn{1}{c}{Orbit representative} & \multicolumn{1}{c}{Stabiliser} & \dim\mathcal{O}_\alpha & \multicolumn{1}{c}{Equations for orbits}\\
                             &
                               \multicolumn{1}{c}{$\alpha=\sM(q,\tilde  q, \bd, \p,\bj)$} & \multicolumn{1}{c}{$B_\alpha$}& \\ \midrule \rowcolor{blue!7}
      1& \sM(q_0,\tilde q_0,\bzero,\bzero,\bzero) & \{\sfb(\varphi,\tilde\varphi,\bzero,\bzero,R)\} & 6 &q=q_0\neq 0, \tfrac1{2 q}(\|\bd\|^2- 2 q \tilde q) = \tilde q_0, q\bj= \bd\times\p\\
      2& \sM(q_0,\tilde q_0,\bzero,\bzero, \ell \be_3) & \{\sfb(\varphi,\tilde\varphi,\bzero,\bzero,R) \mid R\be_3 = \be_3\} & 8 & q=q_0\neq 0, \tfrac1{2 q}(\|\bd\|^2- 2 q\tilde q) = \tilde q_0, \|q\bj - \bd \times\p \| = \ell >0\\\rowcolor{blue!7}
      3& \sM(0,\tilde q_0,\bzero,\bzero,\bzero) & B & 0 & q=0, \tilde q = \tilde q_0, \bd= \p = \bj = \bzero \\
      4& \sM(0,\tilde q_0,\bzero,\bzero,\ell\be_3) & \{\sfb(\varphi,\tilde\varphi,\bbeta,\ba,R) \mid R\be_3 = \be_3\} & 2 & q=0, \tilde q = \tilde q_0, \bd = \p = \bzero, \|\bj\|=\ell>0 \\\rowcolor{blue!7}
      5& \sM(0,\tilde q_0,\bzero, p\be_3, h\be_3) & \{\sfb(\varphi,\tilde\varphi,\bbeta,a \be_3,R) \mid R\be_3 = \be_3\} & 4 & q=0, \tilde q= \tilde q_0, \bd = \bzero, \|\p\|= p>0, \bj \cdot \p = h p\\
      6& \sM(0,0,d \be_3,\bzero,\bzero) & \{\sfb(\varphi,0,\beta \be_3,\ba,R) \mid R\be_3 = \be_3, \ba \cdot \be_3 = 0 \} & 6 & q=0, \|\bd\|=d>0, \bd \times \p = \bzero\\\rowcolor{blue!7}
      7& \sM(0,0,d\be_3, p\be_2,\bzero) & \{\sfb(\varphi,0,\beta \be_3, a \be_2,I)\} & 8 & q = 0, \|\bd\| = d >0, \|\bd \times \p\| = dp >0\\
      \bottomrule
    \end{tabular}
  }
  \caption*{This table lists the different coadjoint orbits of the
    Bargmann group in planon language.  In each case we exhibit an orbit representative
    $\alpha\in \b^*$, its stabiliser subgroup $B_\alpha$ inside the
    Bargmann group, the dimension $\dim \eO_\alpha$ of the orbit and
    the equations defining the orbit.}
\end{table}

The pull-back to the space of parameters of the Maurer--Cartan
one-form on the Bargmann group (equation (4.1) in
\cite{Figueroa-OFarrill:2024ocf}) is now given by
\begin{equation}
  \label{eq:MC-one-form-Bargmann}
  b^{-1}db = \sA(d\varphi - \ba^T d\bbeta-\tfrac12
  \|\ba\|^2d\tilde\varphi, d\tilde\varphi, R^T(d\bbeta +
  \ba d\tilde\varphi), R^T d\ba, R^TdR)
\end{equation}
where $b = \sfb(\varphi,\tilde\varphi, \bbeta,\ba,R)$, and
its contraction with the moment $\sM(q,\tilde q, \bd,\p,J)$
(equation (4.2) in \cite{Figueroa-OFarrill:2024ocf}) is now given by
\begin{multline}
  \label{eq:MC-one-form-contracted-Bargmann}
  \left<\sM(q,\tilde q, \bd,\p,J),b^{-1}db\right> =
  qd\varphi - (\tilde q + \tfrac12 q \|\ba\|^2 +
  (R\bbeta)^T\ba)d\tilde\varphi\\ - (R \bd + q
  \ba)^T d\bbeta + (R\p)^T d\ba + \tfrac12 \Tr J^T R^T d R.
\end{multline}

In the special case of $n=3$ dimensions, and using the isomorphism
$\varepsilon: \RR^3 \to \so(3)$ defined by
$\varepsilon(\ba)\boldsymbol{b} = \ba \times \boldsymbol{b}$, which
allows us to write $J = \varepsilon(\bj)$, we have that
\begin{multline}
  \label{eq:M-C-form-contracted-3d}
  \left<\sM(q,\tilde q, \bd,\p,\bj),b^{-1}db\right> =
  qd\varphi - (\tilde q + \tfrac12 q \|\ba\|^2 +
  R\bbeta\cdot\ba)d\tilde\varphi\\ - (R \bd + q
  \ba)\cdot d\bbeta + (R\p)\cdot d\ba + \bj \cdot
  \varepsilon^{-1}(R^{-1}dR),
\end{multline}
which is simply the translation of equation (5.2) in
\cite{Figueroa-OFarrill:2024ocf}.

\section{Central extensions}
\label{app:extensions}

In this Appendix we collect a result necessary in
Section~\ref{sec:coadjoint-orbits-n=3} for the classification of
coadjoint orbits of the centrally extended planon group.

Let $\g$ be a finite-dimensional (for definiteness) real Lie algebra and
let $H^2(\g)$ denote the second Chevalley--Eilenberg cohomology with
values in the trivial representation.  This is a finite-dimensional
real vector space with basis $[\omega_i]$, for $i=1,\dots,\dim
H^2(\g)$, where $[\omega_i]$ is the cohomology class of the
$2$-cocycle $\omega_i \in \wedge^2 \g^*$.  This allows us to define a
central extension
\begin{equation}
  \begin{tikzcd}
    0 \ar[r] & H^2(\g) \ar [r] & \widehat \g \ar[r] & \g \ar[r] & 0,
  \end{tikzcd}
\end{equation}
whose brackets are given explicitly by
\begin{equation}
  [X , Y]_{\widehat\g} = [X,Y]_\g + \sum_{i=1}^{\dim H^2(\g)}
  \omega_i(X,Y) Z_i \qquad\text{and}\qquad [Z_i,-]_{\widehat\g} = 0,
\end{equation}
for all $X,Y \in \g$.  (We identify $\g$ with a subspace of
$\widehat\g$, so in effect $\widehat\g  = \g \oplus \bigoplus_i \RR
Z_i$, where the above brackets live.)

Now consider a one-dimensional central extension $\widetilde\g$:
\begin{equation}
  \begin{tikzcd}
    0 \ar[r] & \RR Z \ar [r] & \widetilde \g \ar[r] & \g \ar[r] & 0,
  \end{tikzcd}
\end{equation}
with brackets
\begin{equation}
  [X , Y]_{\widetilde\g} = [X,Y]_\g + \omega(X,Y) Z \qquad\text{and}\qquad [Z,-]_{\widetilde\g} = 0,
\end{equation}
for all $X,Y \in \g$ and where $\omega \in \wedge^2 \g^*$ is the
corresponding $2$-cocycle.  In cohomology, $[\omega] = \sum_i c_i
[\omega_i]$, since the $[\omega_i]$ span $H^2(\g)$.  We can therefore
choose $\omega$, perhaps by modifying it by a coboundary, so that
$\omega = \sum c_i \omega_i$ as cocycles.

The claim is that there exists a Lie algebra surjective homomorphism
$\varphi: \widehat\g \to \widetilde\g$ which is the identity on $\g$;
explicitly:
\begin{equation}
  \varphi(X) = X \qquad\text{and}\qquad \varphi(Z_i) = c_i Z.
\end{equation}
It follows that
\begin{align*}
  [\varphi(X),\varphi(Y)]_{\widetilde \g} &= [X,Y]_{\widetilde\g} \\
                                          &= [X,Y]_\g + \omega(X,Y) Z\\
                                          &= [X,Y] _\g + \sum_i c_i \omega_i(X,Y) Z\\
                                          &= \varphi\left([X,Y]_\g + \sum_i \omega_i(X,Y) Z_i\right)\\
                                          &= \varphi\left( [X,Y]_{\widehat\g} \right),
\end{align*}
and hence, since $Z_i$ and $Z$ are central in their respective Lie
algebras, that $\varphi$ is a Lie algebra homomorphism, which is
moreover manifestly surjective.  The kernel of $\varphi$ consists of
those $\sum_i a_i Z_i$ such that $\sum_i a_i c_i = 0$ and hence spans
a hyperplane in $H^2(\g)$.  Conversely, every hyperplane in $H^2(\g)$
defines a one-dimensional central extension of $\g$.

\section{The extended planon group is a regular semidirect product}
\label{sec:regul-centr-extend}

The method of Mackey for constructing UIRs of a semidirect product $B
\ltimes A$, for $A$ abelian (or more generally $A$ nilpotent) is
guaranteed to produce all UIRs whenever the semidirect product is
\emph{regular}.  Let $G = B \ltimes A$ be a semidirect product and let
$\widehat A$ denote the unitary dual: the space of UIRs of $A$.  Then
the semidirect product $G = B \ltimes A$ is said to be
\textbf{regular} if there exist a countable family $\{\eB_i\}$ of
$G$-invariant Borel subsets $\eB_i \subset \widehat A$, so that every
$G$-orbit $\eO \subset \widehat A$ is the intersection $\bigcap_i
\eB_{n_i}$ of some subfamily of the $\{\eB_i\}$.

In this appendix we verify that the description of the centrally
extended planon group $G$ as a semidirect product $B \ltimes A$, with
$B$ the Bargmann group and $A$ a two-dimensional abelian group, is
regular.  Since $A$ is abelian, $\widehat A$ can be identified with
the dual of the Lie algebra $\a$ of $A$, since to every
$\alpha \in \a^*$ we may associate the unitary one-dimensional
representation with character $\chi(\exp(X)) = e^{i\alpha(X)}$, for
$X \in A$.

In the case of interest, $G$ is the centrally extended planon
group, $B$ is the Bargmann group and $A$ the two-dimensional abelian
group with Lie algebra spanned by $H,W$.  The unitary dual $\widehat A$ is
the dual of the Lie algebra and this is a copy of the plane with
coordinates $(E,c)$.  As we saw in Section~\ref{sec:quant-plan-part},
the action of the planon group on $\widehat A$, sends $(E,c) \mapsto (E +
c \tilde\varphi, c)$ and hence there are two classes of orbits:
point-like orbits $\eO_{(E,0)} = \{(E,0)\}$ and one-dimensional orbits
$\eO_{(0,c)} = \RR \times \{c\}$ for every $c \neq 0$.

Consider the following countable family of Borel subsets of the plane:
\begin{equation}
  \eB_E(r,s) := \bigcup_{r<E<s} \eO_{(E,0)} \qquad\text{and}\qquad
  \eB_c(r,s) := \bigcup_{r<c<s} \eO_{(0,c)},
\end{equation}
where $r,s \in \QQ$ are rational numbers.  The subset $\eB_E(r,s)$
consists of those points $(E,0)$ with $r < E < s$ and hence it is an
open interval in the $c=0$ axis.  The subset $\eB_c(r,s)$ consists of
those points $(E,c)$ with $r < c < s$ and $E$ arbitrary and hence it
is an open (infinite) vertical strip.

It is clear that $\eO_{(E,0)}$ lies in every $\eB_E(r,s)$ and hence it
lies in their intersection $\bigcap_{r<E<s} \eB_E(r,s)$.  Furthermore,
if $E-E'\neq 0$, there will be rational
numbers $r,s$ with $r < E < s$ such that $E'\not\in(r,s)$ and hence 
$\eO_{(E',0)} \not\subset \bigcap_{r<E<s} \eB_E(r,s)$.  Therefore,
$\bigcap_{r<E<s} \eB_E(r,s) = \eO_{(E,0)}$.

The argument for $\eO_{(0,c)}$ is identical.  The orbit $\eO_{(0,c)}$
belongs to every $\eB_c(r,s)$ and hence it also belongs to the
intersection $\bigcap_{r<E<s} \eB_c(r,s)$.  Moreover if $c' \neq c$,
there will be some rational numbers $r,s$ with $r < c < s$ such that
$c' \not\in (r,s)$ and hence $\eO_{(0,c')} \not\subset \bigcap_{r<E<s}
\eB_c(r,s)$.  Therefore $\bigcap_{r<c<s} \eB_c(r,s) = \eO_{(0,c)}$.

In summary, every orbit lies in the intersection of a subfamily of the
countable family
\begin{equation}
  \left\{ \eB_E(r,s)  ~\middle |~ r < E < s~\text{with}~r,s \in
    \QQ\right\}\cup \left\{ \eB_c(r,s)  ~\middle |~ r < c < s~\text{with}~r,s \in
    \QQ\right\}
\end{equation}
of Borel sets of $\a^*$, and hence the semidirect product $B \ltimes
A$ is regular.

\bibliographystyle{utphys}
\bibliography{bibl}

\end{document}